\documentclass[11pt,letterpaper]{article}
\usepackage[margin=1in]{geometry}
\usepackage[T1]{fontenc}
\usepackage{amsmath,amsthm,amssymb,amsfonts}
\usepackage{xspace}
\usepackage{url}
\usepackage{framed}
\usepackage{hhline}
\usepackage{pifont}
\usepackage{hyperref}
\usepackage[dvipsnames]{xcolor}
\usepackage[ruled,linesnumbered,noend]{algorithm2e}
\usepackage{appendix}
\usepackage{babel}
\usepackage{enumitem}

\RequirePackage[type1,tt=false]{libertine}

\usepackage{multirow}
\usepackage{tablefootnote}
\usepackage{afterpage}
\usepackage{lipsum}
\usepackage{longtable}
\usepackage{booktabs}
\usepackage{thm-restate}
\usepackage{graphicx}
\usepackage{subcaption}
\usepackage[normalem]{ulem}

\usepackage{bm}
\usepackage{bbm}
\usepackage{mathtools}
\usepackage{mathrsfs}
\usepackage[capitalise,noabbrev]{cleveref}
\usepackage{natbib}
\usepackage{microtype}
\usepackage{comment}
\usepackage{tikz}
\usetikzlibrary{positioning, arrows.meta, fit, backgrounds, shapes, arrows.meta, positioning, shapes, fit, backgrounds, calc}
\hypersetup{colorlinks=true,allcolors=blue}

\theoremstyle{plain}
\newtheorem{theorem}{Theorem}[section]
\newtheorem{lemma}[theorem]{Lemma}

\newtheorem{problem}{Problem}

\newtheorem{proposition}[theorem]{Proposition}

\theoremstyle{definition}
\newtheorem{definition}[theorem]{Definition}

\theoremstyle{remark}
\newtheorem{example}[theorem]{Example}
\newtheorem{remark}[theorem]{Remark}
\newtheorem*{remark*}{Remark}

\allowdisplaybreaks[4]

\makeatletter
\newcounter{HALG@line}
\renewcommand{\theHALG@line}{\thealgorithm.\arabic{ALG@line}}
\makeatother

\newcommand{\efpo}{{\sc EFPO}\xspace}
\newcommand{\efmsw}{{\sc EFMSW}\xspace}
\DeclareMathOperator{\sigma2}{\Sigma_2^p}
\DeclareMathOperator{\pi2}{\Pi_2^p}
\DeclareMathOperator{\theta2}{\Theta_2^p}
\DeclareMathOperator{\d2}{D^p}
\DeclareMathOperator{\delta2}{\Delta_2^p}
\DeclareMathOperator{\classP}{P}
\DeclareMathOperator{\classNP}{NP}
\DeclareMathOperator{\classcoNP}{coNP}
\DeclareMathOperator{\cdott}{\,\cdot\,}

\newcommand{\maxoddSAT}{{\sc MaxOddSAT}\xspace}
\newcommand{\sat}{{\sc SAT}\xspace}
\newcommand{\unsat}{{\sc UNSAT}\xspace}
\newcommand{\aeCNF}{$\forall\exists$-\textsc{3CNF}\xspace}
\newcommand{\easubsetsum}{$\exists\forall$-\textsc{SubsetSum}\xspace}
\newcommand{\vcmember}{{\sc VertexCoverMember}\xspace}
\newcommand{\partition}{{\sc Partition}\xspace}
\newcommand{\xc}{{\sc X3C}\xspace}
\newcommand{\eaxc}{{$\exists\forall$-\sc{X3C}}\xspace}
\newcommand{\krs}{{$\kappa$\sc{RegularSubgraph}}\xspace}
\newcommand{\eakrs}{{$\exists\forall$-$\kappa$\sc{RegularSubgraph}}\xspace}
\newcommand{\neeakrs}{{\sc{NonEmpty}$\exists\forall$-$\kappa$\sc{RegularSubgraph}}\xspace}
\newcommand{\boundefpo}{{\sc{Bounded-EFPO}}\xspace}
\newcommand{\yes}{{\sc Yes}\xspace}
\newcommand{\no}{{\sc No}\xspace}
\newcommand{\add}{{\sf add}\xspace}
\newcommand{\mono}{{\sf mono}\xspace}

\begin{document}

\title{Fair and Efficient Allocations: Decision Problems in the Gap of Polynomial Hierarchy}

\author{Xiaolin Bu\\
lin\_bu@sjtu.edu.cn\\
Shanghai Jiao Tong University
\and
Biaoshuai Tao\\
bstao@sjtu.edu.cn\\
Shanghai Jiao Tong University}
\date{}

\maketitle

\begin{abstract}
We consider the fair division problem with indivisible goods and study the following decision problem: given a fair division instance, does there exist an allocation that is envy-free and efficient?
We consider two efficiency criteria: Pareto-optimality and social welfare optimality.
We provide a complete landscape on the computational complexity of this decision problem, with the number of agents ranging from $2$ to $\infty$,
both additive valuations and general valuations, and the more restricted class of $k$-ary valuation functions (where an item's marginal value is restricted to $\{0,1,\ldots,k-1\}$ for some constant $k\geq2$).

One interesting observation is that many versions of the above-mentioned decision problems fall into the ``gap'' between the first and the second levels of the polynomial hierarchy.
Specifically, assuming the polynomial hierarchy does not collapse to the first level (i.e., assuming $\text{NP}\neq\text{coNP}$), these problems are in $(\Sigma_2^{\text{p}}\cap\Pi_2^{\text{p}})\setminus(\text{NP}\cup\text{coNP})$.
In particular, we provide a fine-grained complexity analysis across different parameter regimes, including the number of agents and the choice of valuation models.
Depending on different parameters, many problems admit different complexity classifications, ranging from the intermediate classes $\Theta_2^{\text{p}}$ and $\Delta_2^{\text{p}}$ between the two levels to $\Sigma_2^{\text{p}}$-completeness.

Finally, De Keijzer et al. show the $\Sigma_2^{\text{p}}$-completeness of the decision problem when considering Pareto-optimality as the efficiency criterion with additive valuations.
Our main results extend this result to more restricted settings, such as instances with a constant number of agents or $3$-ary valuation functions, which resolves the open problem given by Bouveret and Lang.
\end{abstract}

\section{Introduction}
We study the fair division problem with indivisible goods: a set of $m$ items/goods is allocated among a set of $n$ agents with heterogeneous preferences.
The study of the fair division problem dates back to~\citet{Steinhaus48,Steinhaus49}, and it has received significant attention among mathematicians, computer scientists, and economists (see the survey~\citep{amanatidis2023fair,procaccia2013cake}).
\emph{Fairness} and \emph{efficiency} are two major measurements on the qualities of allocations.

Among all fairness notions, \emph{envy-freeness}~\citep{gamow1958puzzle} is arguably the most studied one.
It requires that each agent weakly prefers her own allocated bundle to the bundle received by anyone else.
In other words, there is no envy between any pair of agents.
Clearly, an envy-free allocation may not exist: for example, this can easily happen when $m<n$.
A natural and application-driven question is to decide if an envy-free allocation exists for a given fair division instance.
This problem is clearly in the complexity class $\classNP$, as the existence of an envy-free allocation can be certified by an envy-free allocation.
In many settings, this problem is $\classNP$-complete.

For the notions of efficiency, we study \emph{Pareto-optimality} and \emph{social welfare optimality} that are commonly considered in the previous literature.
An allocation is Pareto-optimal if there is no other allocation where every agent's utility is weakly better and at least one agent's utility is strictly better.
An allocation is social welfare optimal if it maximizes the (utilitarian) social welfare, defined as the sum of all agents' utilities.
Clearly, social welfare optimality is a stronger requirement than Pareto-optimality.
In terms of social welfare, there is a total order on all allocations: for any two allocations, we can always compare their social welfare and see which one is more efficient.
In terms of Pareto-optimality, this becomes a partial order: it is possible that two allocations are incomparable in that neither one Pareto-dominates the other.
Nevertheless, an allocation satisfying any of the two efficiency criteria is guaranteed to exist: the allocation space is finite, and we can always find the most efficient one (or the locally most efficient one in the case of Pareto-optimality).
However, considering efficiency without fairness can be problematic: in many scenarios, allocating all items to a single agent can be an efficient allocation, but this allocation has a very poor performance in fairness and is certainly undesirable.
Therefore, it is natural to ask the following question.

\begin{itemize}
    \item[] \emph{Given a fair division instance, does there exist an allocation that is both fair and efficient?}
\end{itemize}
In this paper, we study the computational complexity of the problem above.

Although the existence of an efficient allocation is understood, checking if a given allocation satisfies the efficient notions is not necessarily easy.
This problem is clearly in $\classcoNP$, as an inefficient allocation is certified by a more efficient one.
For Pareto-optimality, this problem is $\classcoNP$-complete even for very restrictive preference functions of the agents~\citep{de2009complexity,aziz2019efficient}, and even for two agents (see Lemma~\ref{lem:decidingPO2agents}).
For social welfare, checking if an allocation maximizes it is easy for additive valuations, as we only need to check if each item is given to the agent who values it the most.
However, it can easily become $\classcoNP$-complete for general monotone valuations that are not necessarily additive, as it can be hard to see if there is a better allocation with a larger social welfare.

Therefore, the hardness of the decision problem we are studying can come from two aspects: deciding the existence of an envy-free allocation, and checking if this is the (locally) most efficient one.
As we have seen, the former problem can be $\classNP$-complete and the latter problem can be $\classcoNP$-complete.
On the other hand, this problem is in $\sigma2$: for a \yes instance where a fair and efficient allocation exists, we can find an allocation and see if it is suboptimal to any other allocation.
The existential certificate for $\sigma2$ can be the envy-free allocation, and the for-all certificate can be an allocation that is potentially more efficient.

For Pareto-optimality, celebrated results have shown that with a general number of agents, this decision problem is $\sigma2$-complete for binary monotone valuations~\citep{bouveret2008efficiency} or additive valuations~\citep{de2009complexity}.
\citet{bouveret2008efficiency} leave it as an open problem whether $\sigma2$-completeness continues to hold for more restrictive cases, such as constant numbers of agents or $k$-ary additive valuations.

For social welfare optimality, this problem cannot be $\sigma2$-complete unless $\sigma2=\pi2$, as this problem is also in $\pi2$.
We can in fact swap the orders of the two certificates.
The verifier takes two allocations $A'$ and $A$ as inputs, and outputs \yes if and only if 1) $A$ is envy-free and 2) the social welfare of $A'$ is no more than that of $A$. 
For a \yes instance, it is certainly true that for any allocation $A'$ we can find an envy-free allocation $A$ with a weakly higher social welfare, and in fact this allocation $A$ need not even depend on $A'$.
For a \no instance, $A'$ can be set to an allocation with the maximum social welfare, and any envy-free allocation $A$ has social welfare strictly smaller than that of $A'$.\footnote{It can be easily proved that the problem is in the complexity class $\text{S}_2^{\text{p}}$ by this argument, a complexity class contained in $\sigma2\cap\pi2$. We will however not elaborate on this, as we will show later in Sect.~\ref{sect:complexityclasses} that this problem is in $\delta2$, which is contained in $\text{S}_2^{\text{p}}$.}
In fact,~\citet{bouveret2008efficiency} have shown the problem is $\delta2$-complete for monotone valuations even with two agents.
For the special case of monotone binary valuations, the complexity drops to $\theta2$-complete for a general number of agents (see Proposition 16 in~\citet{bouveret2008efficiency}).

The objective of this paper is to carefully examine the exact computational complexity for the decision problem of deciding the existence of fair and efficient allocations under different settings.

\subsection{Our Results}
We study the following decision problems in this paper:
\begin{enumerate}
    \item Given a fair division instance, does there exist an envy-free and social welfare optimal allocation?
    \item Given a fair division instance, does there exist an envy-free and Pareto-optimal allocation?
    \item Given a fair division instance and an allocation, does it satisfy Pareto-optimality or maximize the social welfare?
\end{enumerate}

We consider two types of valuation functions: the more restrictive additive functions and the more general monotone set functions.
We also consider the special case with a constant number of agents and the special case where the valuation functions are $k$-ary (where an item's marginal value when included in a set of items can only be in $\{0,1,\ldots,k-1\}$ for some constant $k\geq2$).
The results we obtained are shown in Table~\ref{tab:results} and Table~\ref{tab:results2}.

{\small
\begin{table}[t]
    \centering
    \setlength{\tabcolsep}{3pt}
    \renewcommand\arraystretch{1.1}
    \begin{tabular}{|l||ccc|ccc|}
    \hline
       \multirow{2}{*}{\textbf{EF+MSW}}  & \multicolumn{3}{c|}{Additive Valuations} & \multicolumn{3}{c|}{General Monotone Valuations} \\
         & binary & $k$-ary & general & binary & $k$-ary & general\\
    \hline
    $n=2$ & $\classP$ & $\classP$ & $\classNP$-complete & $\theta2$-complete & $\theta2$-complete & $\delta2$-complete\\
    constant $n$ & $\classP$ & $\classP$ & $\classNP$-complete & $\theta2$-complete & $\theta2$-complete & $\delta2$-complete\\
    general $n$ & $\classNP$-complete & $\classNP$-complete & $\classNP$-complete & $\theta2$-complete & $\theta2$-complete & $\delta2$-complete\\
    \hhline{=||======}
    \multirow{2}{*}{\textbf{EF+PO}} & \multicolumn{3}{c|}{Additive Valuations} & \multicolumn{3}{c|}{General Monotone Valuations} \\
         & binary & $k$-ary & general & binary & $k$-ary & general\\
     \hline
    $n=2$ & $\classP$ & $\classP$ & $\classNP$-complete & $\theta2$-complete & $\theta2$-complete & $\sigma2$-complete\\
    constant $n$ & $\classP$ & $\classP$ & $\sigma2$-complete & $\theta2$-complete & $\theta2$-complete & $\sigma2$-complete\\
    general $n$ & $\classNP$-complete & $\sigma2$-complete & $\sigma2$-complete & $\sigma2$-complete & $\sigma2$-complete & $\sigma2$-complete\\
    \hline
    \end{tabular}
    \caption{Results for the first two problems of deciding the existence of fair and efficient allocation (Intuitions are discussed in Sect.~\ref{sect:complexityclasses}, proofs are available in Sect.~\ref{sec:infty0k-po} and Sect.~\ref{sect:formal}, where results can be located in Sect.~\ref{sect:combine})}
    \label{tab:results}
\end{table}

\begin{table}[t]
\centering
\begin{tabular}{|l||ccc|ccc|}
\hline
\multirow{2}{*}{\textbf{MSW Checking}} 
& \multicolumn{3}{c|}{Additive Valuations} 
& \multicolumn{3}{c|}{General Monotone Valuations} \\
& binary & $k$-ary & general & binary & $k$-ary & general\\
\hline

$n=2$
& \multicolumn{3}{c|}{\multirow{3}{*}{
\parbox{3.cm}{\centering
\textbf{}\\
$\classP$\\
}}}
& \multicolumn{3}{c|}{\multirow{3}{*}{
\parbox{3cm}{\centering
\textbf{}\\
$\classcoNP$-complete\\
}}}
\\

constant $n$ & & & & & & \\
general $n$  & & & & & & \\
\hhline{=||======}
    \multirow{2}{*}{\textbf{PO Checking}} & \multicolumn{3}{c|}{Additive Valuations} & \multicolumn{3}{c|}{General Monotone Valuations} \\
         & binary & $k$-ary & general & binary & $k$-ary & general\\
     \hline
    $n=2$ & $\classP$ & $\classP$ & $\classcoNP$-complete &\multicolumn{3}{c|}{\multirow{3}{*}{
\parbox{3cm}{\centering
\textbf{}\\
$\classcoNP$-complete\\
}}}
\\
    constant $n$ & $\classP$ & $\classP$ & $\classcoNP$-complete & & &\\
    general $n$ & $\classP$ & $\classcoNP$-complete & $\classcoNP$-complete & & &\\
    \hline
    \end{tabular}
    \caption{Results for the third problem of efficiency checking (Proofs are available in Sect.~\ref{sec:efficiency-checking-proof})}
    \label{tab:results2}
\end{table}
}

Interestingly, we find that many of our problems are in the gap between the first level of the polynomial hierarchy (i.e., $\classNP$ and $\classcoNP$) and the second level of the polynomial hierarchy (i.e., $\sigma2$ and $\pi2$).
Formally speaking, they are in $(\sigma2\cap\pi2)\setminus(\classNP\cup\classcoNP)$ assuming the polynomial hierarchy does not collapse to the first level (i.e., assuming $\classNP\neq\classcoNP$).
In more detail, we show that many of these problems are complete in the complexity classes $\theta2$ and $\delta2$, both classes are contained in $(\sigma2\cap\pi2)$ and contain $\classNP\cup\classcoNP$.
Since both $\theta2$ and $\delta2$ contain $\classNP\cup\classcoNP$, a problem complete in $\theta2$ or $\delta2$ is both $\classNP$-hard and $\classcoNP$-hard.
Assuming $\classNP\neq\classcoNP$, these problems are in the gap $(\sigma2\cap\pi2)\setminus(\classNP\cup\classcoNP)$  between the first two levels of the polynomial hierarchy.

Compared to previous results, our paper provides a more fine-grained complexity analysis by considering more restrictive settings with respect to the number of agents and types of values.
It is an interesting phenomenon that different combinations of the studied restrictions lie on (or in between) different levels of the polynomial hierarchy.
On the other hand, we discover that different combinations of restrictions yield problems that have clear features of typical complete problems in complexity classes, which provide many insights that help understand these complexity classes.

Finally,~\citet{bouveret2008efficiency} leave the open problems on the complexity of deciding the existence of an envy-free and Pareto-optimal allocation for additive valuations, or even more restrictive $k$-ary additive valuations. \citet{de2009complexity} answer the former one by showing that the problem remains $\sigma2$-complete for additive valuations.
We extend this result and resolve the latter one by showing that $\sigma2$-completeness continues to hold for $3$-ary additive valuations with a general number of agents, or for additive valuations with a constant number of agents.

We will go through the above-mentioned complexity classes in Sect.~\ref{sect:complexityclasses}.
We will also provide insight for all results in Table~\ref{tab:results} and Table~\ref{tab:results2} by using these fair division problems as running examples for these complexity classes, and we will see why these problems are typical problems that are complete in these classes.
Although the proofs for all the results are available in Sect.~\ref{sec:infty0k-po}, Sect.~\ref{sect:formal}, and Sect.~\ref{sec:efficiency-checking-proof}, we expect Sect.~\ref{sect:complexityclasses} provides sufficient intuitions that serve as proof sketches.

\subsection{Related Work}

\paragraph{Computational Complexity of Maximizing Social Welfare.}
An allocation with the maximum social welfare can be computed in polynomial time for additive valuations, where each item is allocated to the agent with the highest value to it.
However, the problem becomes hard when deciding whether a social welfare-maximizing allocation is compatible with fairness guarantees.
For monotone valuations,~\citet{bouveret2008efficiency} have shown that deciding the existence of an envy-free allocation with maximum social welfare is $\delta2$-complete for the two-agent case, and $\theta2$-complete for a general number of agents with monotone dichotomous valuations, which is a special case of monotone binary valuations.
For general additive valuations,~\citet{aziz2023computing} show that deciding the existence of an EF1 or EFX allocation with maximum social welfare is $\classNP$-complete, except for the two-agent case with EF1 guarantee (which is in $\classP$).

Another line of work considers fair and efficient allocations from the aspect of price of fairness, which denotes the loss in social welfare when fairness is required~\citep{barman2020optimal,bei2021price,li2024complete}, or the complexity and approximability of maximizing social welfare under fair constraints~\citep{aziz2023computing,bu2025approximability,barman2019fair,barman2020optimal}.

\paragraph{Computational Complexity of Pareto-Optimality.}
Given an allocation, deciding whether it is Pareto-optimal is $\classcoNP$-complete, even for tri-valued additive valuations (where each agent's value to each item belongs to $\{p,q,r\}$)~\citep{de2009complexity,aziz2019efficient}.
For the special case with bi-valued additive valuations (where each agent's value to each item belongs to $\{p,q\}$), Pareto-optimality can be decided in polynomial time~\citep{aziz2019efficient}.
When considering fairness, deciding the existence of a Pareto-optimal and envy-free allocation for a general number of agents is $\sigma2$-complete under binary monotone valuations~\citep{bouveret2008efficiency}, or under general additive valuations~\citep{de2009complexity}.

For weaker fairness notions, it is known that Pareto-optimality is compatible with envy-freeness up to one item (EF1) for additive valuations by maximizing Nash social welfare~\citep{caragiannis2019unreasonable}, and can be computed in pseudo-polynomial time~\citep{barman2018finding}.
For $k$-ary valuations, such an allocation is polynomial-time computable~\citep{garg2024computing}.
For envy-freeness up to any item (EFX), Pareto-optimality is incompatible with EFX, and deciding the existence of a Pareto-optimal and EFX allocation is $\classNP$-hard~\citep{garg2023computing}.
Positive results are shown for binary or bi-valued additive valuations, where such an allocation exists and can be computed in polynomial time~\citep{halpern2020fair,babaioff2021fair,amanatidis2021maximum,garg2023computing}.

\paragraph{Other Aspects in Fair Division.}
As envy-free allocations may not exist, a variety of work considers relaxations of envy-freeness, including envy-freeness up to one item (EF1)~\citep{budish2011combinatorial,lipton2004approximately}, envy-freeness up to any item (EFX)~\citep{caragiannis2019unreasonable,plaut2020almost,chaudhury2020efx,berger2022almost,amanatidis2021maximum}, etc.
We refer the readers to the survey~\citep{aziz2022algorithmic} for further details.

\section{Preliminaries}
\subsection{Fair Division Model}
A set $M$ of $m$ goods/items is to be allocated to a set of $n$ agents $N=\{1,\ldots,n\}$.
For notational convenience, we reserve the symbols $n$ and $m$ respectively for the numbers of agents and items \emph{only in this section}, to allow the freedom of using $n$ and/or $m$ to refer to the parameters of the hard problems from which our problems are reduced in the later sections.

Each agent $i$ has a \emph{valuation function} $v_i(\cdot)$ that takes a subset of items $S$ as input and outputs a non-negative value representing agent $i$'s value on the bundle $S$, and we assume $v_i$ satisfies (i) \emph{normalization}: $v_i(\emptyset)=0$ and (ii) \emph{monotonicity}: $v_i(S)\geq v_i(T)$ if $T\subseteq S\subseteq M$.
A valuation function $v_i(\cdot)$ is \emph{additive} if agent $i$ has a non-negative value $v_{ig}$ on each item $g\in M$, and $v_i(S)=\sum_{g\in S}v_{ig}$.
Given a positive integer $k\geq2$, we say that a valuation function is \emph{$k$-ary} if $v_i(S\cup\{g\})-v_i(S)\in\{0,1,\ldots,k-1\}$ for any $S\subseteq M$ and $g\in M$.
Notice that this would imply $v_{ig}\in\{0,1,\ldots,k-1\}$ for any $g\in M$ for additive valuation function $v_i(\cdot)$.
A valuation function is \emph{binary} if it is $k$-ary with $k=2$.

We also make the following natural assumptions on $v_i(\cdot)$ for the convenience of studying computational complexity:
\begin{enumerate}
    \item[(i)] (Integrality): $v_i(S)\in\mathbb{Z}_{\geq0}$ for any $S\subseteq M$;
    \item[(ii)] (Poly-time computability): $v_i(S)$ can be computed in polynomial time with respect to $m$ and $n$.\footnote{It is tempting to make this assumption more general by assuming an oracle access to $v_i(\cdot)$. However, this will be problematic when studying complexity, as we can do many extra undesirable things with this oracle. For example, if the item set $M$ represents the vertex set in an undirected graph $G_i$ and $v_i(S)$ is defined by the size of the maximum clique in the subgraph induced by $S$, an oracle access $v_i(\cdot)$ is (at least weakly) more powerful than an $\classNP$ oracle. As we will see in Sect.~\ref{sect:complexityclasses}, we can solve all problems in the complexity class $\delta2$ in polynomial time. This is certainly not what we want.}
    It is implied the length of the binary string representing the maximum possible value $v_i(M)$ is bounded by a polynomial of $m$ and $n$. 
\end{enumerate}

A fair division instance is given by $(N,M,\{v_i\}_{i\in N})$.
An \emph{allocation} is an ordered partition $(A_1,\ldots,A_n)$ of $M$ where $A_i$ is the set of items allocated to agent $i$.
We will interchangeably call $A_i$ the \emph{bundle} received by agent $i$.
We require all items to be allocated.
However, if we consider the same class of problems studied in this paper with partial allocations allowed, as we will remark in Remark~\ref{rmk:partial}, all our results in the paper continue to hold. As a quick overview of this, when we are performing a reduction, we will always make sure that 1) for a \yes instance, a valid complete allocation exists, and 2) for a \no instance, no valid allocation exists, even allowing partial allocations.

\begin{definition}[Envy-freeness]
    Given a fair division instance $(N,M,\{v_i\}_{i\in N})$, an allocation $(A_1,\ldots,A_n)$ is \emph{envy-free} if $v_i(A_i)\geq v_i(A_j)$ for any pair of agents $i$ and $j$.
\end{definition}

In words, an allocation is envy-free if, for any pair of agents $i$ and $j$, agent $i$ believes her allocated bundle is weakly more valuable than agent $j$'s bundle.
We say that agent $i$ envies agent $j$ if $v_i(A_i)<v_i(A_j)$.

\begin{definition}
    Given a fair division instance $(N,M,\{v_i\}_{i\in N})$, an allocation $(A_1',\ldots,A_n')$ \emph{Pareto-dominates} $(A_1,\ldots,A_n)$ if all the $n$ inequalities $\{v_i(A_i')\geq v_i(A_i)\}_{i\in N}$ hold and at least one of them is strict.
    An allocation $(A_1,\ldots,A_n)$ is \emph{Pareto-optimal} if it is not Pareto-dominated by any allocation.
\end{definition}

\begin{definition}
    Given a fair division instance $(N,M,\{v_i\}_{i\in N})$, the \emph{social welfare} of an allocation $(A_1,\ldots,A_n)$ is defined by $\sum_{i=1}^nv_i(A_i)$. An allocation $(A_1,\ldots,A_n)$ is \emph{social welfare optimal} if it has the maximum social welfare among all allocations.
\end{definition}
It is straightforward to see that a social welfare optimal allocation is always Pareto-optimal.

We also need the following technical tool.
\begin{definition}\label{def:truthtable}
    Given a fair division instance $(N,M,\{v_i\}_{i\in N})$, the \emph{valuation truth table}, denoted by $\mathcal{T}(N,M,\{v_i\}_{i\in N})$, or by just $\mathcal{T}$ when there is no ambiguity on the instance considering, is a table of size $(T+1)^{n^2}$ where $T=\max_{i\in N}v_i(M)$.
    A cell in the table is specified/indexed by a $n\times n$ matrix $[u_{ij}]_{i\in N,j\in N}\in\{0,1,\ldots,T\}^{n\times n}$.
    The value stored in this cell is ``true'' if there is an allocation $(A_1,\ldots,A_n)$ such that $v_i(A_j)=u_{ij}$ for every $i,j\in N$, and it is ``false'' otherwise.
\end{definition}

Notice that the size of the table is polynomial in $m$ and $n$ if $n$ is a constant and valuations are $k$-ary.
It does not matter if the valuation functions are additive or not, as $T\leq mk$ holds for $k$-ary valuations.

\subsection{\efpo and \efmsw Problems}
In this section, we formally define the two decision problems \efpo and \efmsw studied in this paper.
Both decision problems' instances are fair division instances $(N,M,\{v_i\}_{i\in N})$.
The \yes instances of \efpo are those admitting envy-free and Pareto-optimal allocations; for \efmsw, those admitting envy-free and social welfare optimal allocations.

We consider both problems with various special cases, such as a fixed number of agents, additive valuations, and/or $k$-ary valuations.
For \efpo, we use three parameters to specify the exact problem/special case, and we will write \efpo-[\ding{172},\ding{173},\ding{174}], illustrated as follows:
\begin{itemize}[leftmargin=*]
    \item The first parameter \ding{172} specifies the maximum number of agents, and it takes a value from $\{2,3,\ldots,\}\cup\{\infty\}$, where $\infty$ indicates that the number of agents can be arbitrary (depending on the input instance) and other values indicate that the number of agents is bounded by a fixed/constant number;
    \item The second parameter \ding{173} takes a value of either \add or \mono, where \add indicates all valuations are additive and \mono indicates the valuations can be general monotone functions;
    \item The third parameter \ding{174} takes a value from $\{2,3,\ldots\}\cup\{\infty\}$, where a value $k\in\{2,3,\ldots\}$ indicates that all valuation functions are $k$-ary and $\infty$ indicates that this type of assumptions is not made.
\end{itemize}

\begin{example}
\efpo-$[2,\mono,\infty]$ refers to the decision problem \efpo where there are two agents in every instance, the valuation functions need not to be additive, and no $k$-ary assumption is made; \efpo-[$5,\add,2$] refers to the decision problem \efpo with at most $5$ agents and additive binary valuations.     
\end{example}

We define \efmsw-[$\cdot,\cdot,\cdot$] similarly.
Let $\add \le \mono$ and $x \le \infty$ for any $x\in \mathbb{Z}^+$.
Then for any $a,b,c,a',b',c'$ with $a\leq a'$, $b\leq b'$, and $c\leq c'$, \efpo-[$a,b,c$] is a special case of \efpo-[$a',b',c'$] (meaning that an \efpo-[$a,b,c$] instance is also an \efpo-[$a',b',c'$] instance), and is ``no harder than'' \efpo-[$a',b',c'$].
The same holds for \efmsw.

\section{Complexity Classes between First Two Levels of Polynomial Hierarchy; How \efpo and \efmsw Are Placed}
\label{sect:complexityclasses}
In this section, we introduce the complexity classes $\theta2,\delta2$ between the first level of the polynomial hierarchy ($\classNP$ and $\classcoNP$) and the second level of the polynomial hierarchy ($\sigma2$ and $\pi2$).
We will use the fair division problems \efmsw-[$\cdot,\cdot,\cdot$] and \efpo-[$\cdot,\cdot,\cdot$] as running examples to explain these complexity classes, as we find that these fair division problems are very typical examples that help better understand these complexity classes.
The discussion in this section also provides significant intuitions for the results in Table~\ref{tab:results}, and we expect that the readers can mostly understand how the results in Table~\ref{tab:results} come from by reading this section.
The descriptions in this section are not intended to be formal.
Formal treatments of the complexity concepts can be found in many textbooks.

\subsection{Complexity Theory Basics and Overview}
We begin by briefly going through basic complexity concepts and defining the well-known complexity classes $\classP,\classNP,\classcoNP,\sigma2,$ and $\pi2$.

Let $\Sigma$ be an \emph{alphabet-set}, which is typically set to binary $\Sigma=\{0,1\}$, and let $\Sigma^\ast$ be the set of strings with arbitrary lengths using the alphabet-set $\Sigma$. 
A \emph{language}, or a \emph{decision problem}, is a subset $L$ of $\Sigma^\ast$.
A \yes instance of $L$ is a string $x$ with $x\in L$, and a \no instance of $L$ is a string $x$ with $x\notin L$.
For example, for \efpo, $L$ contains all the strings that encode fair division instances in each of which an envy-free and Pareto-optimal allocation can be found.
When referring to the fair division problems, we assume without loss of generality that each string represents a valid fair division instance.
The \emph{complement} of a language $L$, denoted by $\bar{L}$, is defined by $\bar{L}=\{x\mid x\notin L\}$.
Given two languages $L$ and $L'$, we say that $L$ \emph{Karp-reduces to} $L'$ if there is a polynomial-time Turing machine that takes a string $x$ as input and outputs another string $x'$ such that $x\in L$ if and only if $x'\in L'$.
All reductions in this paper are Karp reductions, and we will simply say $L$ reduces to $L'$ from now on.

A \emph{complexity class} $X$ is a set of languages.
A language $L$ is $X$-hard if every language in $X$ reduces to $L$.
A language $L$ is $X$-complete if $L\in X$ and $L$ is $X$-hard.
From this definition of hardness, it is straightforward to see that a $Y$-hard language is also $X$-hard if $X\subseteq Y$.
We assume the readers are familiar with the five complexity classes $\classP,\classNP,\classcoNP,\sigma2,$ and $\pi2$, and we give a quick review of $\sigma2$ and $\pi2$ below.
When we say the running time of a Turing machine is polynomial, it is with respect to the length of $x$.
\begin{itemize}[leftmargin=*]
    \item $\textbf{P}$ is the class of languages that can be decided by a deterministic Turing machine in polynomial time.
    \item $\textbf{NP}$ is the class of languages that can be decided by some nondeterministic polynomial-time Turing machine.
    Equivalently, a language $L\in\classNP$ if and only if there exists a polynomial-time Turing machine $\mathbb{M}$ such that for any input $x$, $x\in L$ if and only if there exists a polynomial-length certificate $u$ such that $\mathbb{M}(x,u)=1$.

    \item $\textbf{coNP}$ is defined by $\{L\mid\bar{L}\in \classNP\}$.
    Equivalently, a language $L\in\classcoNP$ if and only if there exists a polynomial-time Turing machine $\mathbb{M}$ such that for all input $x$, $x\in L$ if and only if for any polynomial-length certificate $u$, we have $\mathbb{M}(x,u)=1$.

    \item $\bm\Sigma_{\bm2}^\textbf{p}$: a language $L\in\sigma2$ if and only if there exists a polynomial-time Turing machine $\mathbb{M}$ such that for any input $x$, $x\in L$ if and only if there exists a polynomial-length certificate $u_1$, such that for any polynomial-length certificate $u_2$, we have $\mathbb{M}(x,u_1,u_2)=1$.

    \item $\bm\Pi_{\bm2}^\textbf{p}$ is defined by $\{L\mid\bar{L}\in\sigma2\}$. Equivalently, a language $L\in\pi2$ if and only if there exists a polynomial-time Turing machine $\mathbb{M}$ such that for any input $x$, $x\in L$ if and only if for any polynomial-length certificate $u_1$, there exists a polynomial-length certificate $u_2$ such that $\mathbb{M}(x,u_1,u_2)=1$.
\end{itemize}

\paragraph{A brief overview of $\theta2$ and $\delta2$.}
Here, we briefly describe the complexity classes $\theta2$ and $\delta2$. We will elaborate on them more in the later subsections, with fair division problems as running examples.

To introduce the complexity classes $\theta2$ and $\delta2$, we need the notion of an \emph{NP oracle}.
Informally, a Turing machine with an NP oracle is a Turing machine that can query whether $x\in L$ and get the answer (\yes or \no) in one step for any language $L\in\classNP$.
We refer the readers to standard complexity textbooks for the precise definition of a Turing machine with an NP oracle.
Notice that an NP oracle can also solve any $\classcoNP$ problem in one step, as we only need to flip the answer.
Thus, an NP oracle has the same power as a coNP oracle.
An equivalent definition for $\sigma2$ is the set of languages that can be decided by a nondeterministic polynomial-time Turing machine with an NP oracle, and we write $\sigma2=\classNP^{\classNP}$.
Correspondingly, $\pi2=\classcoNP^{\classNP}$.

If we consider problems that can be solved by a \emph{deterministic} Turing machine in polynomial time with an NP oracle, this is exactly the complexity class $\delta2$.
\begin{itemize}[leftmargin=*]
\item $\bm\Delta_{\bm2}^\textbf{p}$, defined by $\classP^{\classNP}$, is the class of languages that can be decided by some polynomial-time Turing machine with an NP oracle. 
\end{itemize}
It is clear that $\delta2$ is contained in both $\sigma2$ and $\pi2$: $\delta2\subseteq \sigma2\cap\pi2$.

Notice that, in the definition of $\delta2$, we can apply the NP oracle for up to polynomially many times in an \emph{adaptive} way.
If we draw the decision tree where each node represents a query and each leaf represents a final decision (\yes or \no), the depth of the tree is polynomial in $|x|$ and the total number of nodes is exponential.
Therefore, if we insist on \emph{non-adaptive} queries to the NP oracle, the total number of queries can be exponential. We will elaborate on this more later.

A subset of $\delta2$ is $\theta2$, which is defined as the decision problems that can be solved in polynomial time with polynomially many \emph{non-adaptive} queries to an NP oracle.
\begin{itemize}[leftmargin=*]
\item $\bm\Theta_{\bm2}^\textbf{p}$, defined by $\classP^{\classNP}_{\parallel}$, is the class of languages that can be decided by some polynomial-time oracle machine that non-adaptively applies $\classNP$ oracles for polynomially many times.
\end{itemize}
\citet{hemachandra1987strong} and \citet{buss1991truth} independently realized that $\theta2$ is the same as $\classP^{\classNP[\log ]}$, where $\classP^{\classNP[\log ]}$ is the set of problems that can be decided by a Turing machine in polynomial time with $O(\log(|x|))$ adaptive queries to an NP oracle.
That is, a logarithmic number of adaptive queries has the same power as a polynomial number of non-adaptive queries.
It is easy to see that a logarithmic number of adaptive queries can be simulated by a polynomial number of non-adaptive queries, as the number of nodes in a binary decision tree with a logarithmic depth is polynomial.
The other direction is harder to see.

We can further restrict the number of queries to the NP oracle and get sub-class of $\theta2$.
If the number of queries is bounded by $1$, the complexity class is denoted by $\classP^{\classNP[1]}$, and we have $\classNP\cup\classcoNP\subseteq \classP^{\classNP[1]}$ as any problem in $\classNP$ or $\classcoNP$ can be solved directly with one query of an NP oracle.

Putting together, the relationship between these complexity classes is given below.

\begin{proposition}
$\classP\subseteq\classNP\cap\classcoNP\subseteq\classNP\cup\classcoNP\subseteq\classP^{\classNP[1]}\subseteq\theta2\subseteq\delta2\subseteq\sigma2\cap\pi2$.
\end{proposition}

It is widely believed that $\classNP\neq\classcoNP$; otherwise, the polynomial hierarchy collapses to the first level.
However, as we will show next, if $\classNP\neq\classcoNP$, even this containment $\classNP\cup\classcoNP\subseteq\classP^{\classNP[1]}$ is proper.\footnote{To the best of our knowledge, we did not find such a proof in the previous literature. We therefore provide a proof of this simple fact here.}
To see this, we need the following well-known proposition.

\begin{proposition}
    If $L$ is both $\classNP$-hard and $\classcoNP$-hard, then either $L\in\classNP$ or $L\in\classcoNP$ implies $\classNP=\classcoNP$.
\end{proposition}

With the above proposition, assuming $\classNP\neq\classcoNP$, a language that is both $\classNP$-hard and $\classcoNP$-hard cannot be in $\classNP\cup\classcoNP$.
To show $\classNP\cup\classcoNP\subsetneq\classP^{\classNP[1]}$, it suffices to find $L\in\classP^{\classNP[1]}$ that is both $\classNP$-hard and $\classcoNP$-hard.
The following language does the job:
$$L=\left\{(x,\phi)_{x\in\{0,1\},\phi\mbox{ is a \sat instance}}\mid (x=1\wedge\phi\mbox{ is satisfiable})\vee (x=0\wedge\phi\mbox{ is unsatisfiable})\right\}.$$
This language is clearly in $\classP^{\classNP[1]}$, as we only need one query to decide the satisfiability of $\phi$, and we can output \yes or \no depending on the value of $x$.
It is $\classNP$-hard: the reduction simply maps a \sat instance $\phi$ to $(1,\phi)$.
It is also $\classcoNP$-hard: the reduction simply maps an \unsat instance $\phi$ to $(0,\phi)$.

For the complexity classes $\theta2$ and $\delta2$ that contains $\classNP\cup\classcoNP$, for any language $L$ that is complete in $\theta2$ or $\delta2$, it is also both $\classNP$-hard and $\classcoNP$-hard, and we can conclude that $L$ is outside $\classNP\cup\classcoNP$ (assuming $\classNP\neq\classcoNP$).

We have identified many fair division problems that are complete in $\theta2$ and $\delta2$ (as shown in Table~\ref{tab:results}).
By the analysis above, these problems are in the ``gap'' between the first two levels of polynomial hierarchy, namely, in $(\sigma2\cap\pi2)\setminus(\classNP\cup\classcoNP)$, assuming $\classNP\neq\classcoNP$.

As a final remark, many other complexity classes between the first two levels of the polynomial hierarchy are not discussed in this paper.
Notable examples include 1) $\d2$, which is between $\classP^{\classNP[1]}$ and $\classP^{\classNP[2]}$ and typical problems that are $\d2$-complete includes deciding if the maximum clique of the graph has size exactly $k$, and 2) $\text{S}_2^p$, which contains $\delta2$ and is contained in $\sigma2\cap\pi2$.

\paragraph{On fair division problems.}
All problems with \efmsw-[$\cdot,\cdot,\cdot$] and \efpo-[$\cdot,\cdot,\cdot$] are in $\sigma2$.
To see this, the existential certificate in the definition of $\sigma2$ is an envy-free allocation $A$, and the for-all certificate is a potential allocation $A'$ with a higher social welfare (for \efmsw) or Pareto-dominates $A$ (for \efpo).
The difficulties of these problems can come from two directions: 1) deciding the existence of an envy-free allocation, and 2) deciding if this envy-free allocation is most efficient.
If 2) can be decided in polynomial time, then the problem is in $\classNP$, as we do not require the for-all certificate in this case.
Interestingly, however, if the decision problem 2) is $\classcoNP$-hard, the fair division problem can still be in $\classNP$.
An example is \efpo-$[2,\add,\infty]$ (see Lemmas~\ref{lem:member:20infty-po} and~\ref{lem:decidingPO2agents}).
Therefore, the exact complexity of each fair division problem not only depends on whether 1) or 2) is hard, but also on how the difficulties for 1) and 2) are related.
This will be elaborated in the following sub-sections.

\subsection{Fair Division Problems in $\classP$}
We first identify the easiest problems. These are exactly \efpo-$[n,\add,k]$ and \efmsw-$[n,\add,k]$ for constants $n$ and $k$: those instances with constant numbers of agents and additive $k$-ary valuation functions.

As we have remarked after Definition~\ref{def:truthtable}, for each such instance, the valuation truth table $\mathcal{T}$ has a polynomial size.
If the truth table is constructed, checking the existence of an envy-free and Pareto-optimal allocation can obviously be done in polynomial time:
the index of the cell already tells if the corresponding allocation is envy-free; for each ``envy-free index'' that has value ``true'', it suffices to check the Boolean value of each cell that represents a Pareto-dominating allocation to see if any of them has value ``true''.
As a remark, for each cell that has value ``true'', it may correspond to more than one allocation.

Therefore, the problem is in $\classP$ if $\mathcal{T}$ can be constructed in polynomial time.
This is true for additive valuations with constant numbers of $n$ and $k$: it can be done by the standard dynamic programming method, which has also been applied in other fair division papers~\citep{aziz2023computing,han2023average,bu2025approximability}.

However, the dynamic programming method will fail in any of the two cases below:
\begin{enumerate}
    \item the number of agents $n$ is not a constant, or the valuation function is not $k$-ary (in which case the table $\mathcal{T}$ is not of polynomial size), or
    \item the valuation function is not additive (in which case building a recurrence relation in the dynamic programming is infeasible).
\end{enumerate}

\subsection{Fair Division Problems in $\classNP$}
Typical fair+efficiency problems in $\classNP$ are those such that it is easy to check if an allocation is efficient.
For example, for all problems in \efmsw-$[\cdot,\add,\cdot]$, i.e., all problems in \efmsw with additive valuations, it is straightforward to check if an allocation is social welfare optimal, as an allocation is social welfare optimal if and only if each item is given to the agent who values it most; this establishes $\classNP$ membership.
As discussed in the previous subsection, \efmsw is easy and belongs to $\classP$ for $n$ and $k$ being both constants.
However, it becomes $\classNP$-hard when either $n$ or $k$ is not a constant.
For additive valuations that are not $k$-ary, even for two agents, deciding the existence of envy-free allocations is $\classNP$-hard by a straightforward reduction from the well-known $\classNP$-complete problem \partition.
For a general number of agents, even with binary valuations, this decision problem is $\classNP$-hard due to \citet{aziz2015fair}.

The problem \efpo-$[\cdot,\add,2]$ (i.e., instances with additive binary valuations) is also in $\classNP$, as an allocation is Pareto-optimal if and only if it is \emph{non-wasteful}, meaning that no agent receives an item with value $0$ (recall that each agent's value on each item can only be $0$ or $1$).
In addition, it is $\classNP$-hard due to \citet{aziz2015fair}.

The only exception is \efpo-$[2,\add,\cdot]$ (i.e., instances with additive valuations and only two agents).
Although it is $\classcoNP$-hard to check if an envy-free allocation is Pareto-optimal (see Lemma~\ref{lem:decidingPO2agents}), the problem is still in $\classNP$.
This is due to a property that is unique for two agents and additive valuations: if an allocation is envy-free, then any of its Pareto-improvements is also envy-free (to see this, for two agents, an agent does not envy the other if and only if she receives at least $\frac12$ fraction of the value of the entire item set $M$; any Pareto-improvement preserves this property).
Therefore, we know that it is a \yes instance as long as an envy-free allocation exists.
In addition, this problem is also $\classNP$-complete, as checking the existence of an envy-free allocation is $\classNP$-complete by the straightforward reduction from \partition.

For \efpo, the containment in $\classNP$ does not hold for a general additive cases with more agents.
The problem becomes $\sigma2$-complete.
\citet{de2009complexity} shows that the problem \efpo-$[\infty,\add,\infty]$ is $\sigma2$-complete.
In Sect.~\ref{sect:hardness}, we extend De Keijzer et al.'s result and show that the problem remains $\sigma2$-complete even for a constant number of agents.

\subsection{Complexity Class $\theta2$ and Fair Division Problems Contained in}
As we have mentioned before, $\theta2$ contains those problems that can be solved with a polynomial number of non-adaptive queries to an NP oracle.
Equivalently, a problem in $\theta2$ can be solved with a logarithmic number of adaptive queries.
Intuitively, for those problems that are complete in $\theta2$, there should be sufficiently many leaves in the decision tree that correspond to ``accept'', and there should be sufficiently many leaves that go to ``reject''.
If the number of leaves that represent ``accept'' is small, we only need to check those paths that lead to ``accept'', we can ``compress'' the queries along each path to two NP queries: combining those queries with answers \yes to a single large CNF Boolean formula which is supposed to be satisfiable, which corresponds to one NP query; combining those queries with answers \no to a single large Boolean formula which is supposed to be unsatisfiable\footnote{notice that the combination should use logical operator ``OR'', so the combined formula is not in CNF form, but it is still in $\classcoNP$ to check if this formula is unsatisfiable.}, which corresponds to the other NP query.
Similar things happen if the number of leaves corresponding to ``reject'' is small.
In this way, the number of queries can be made substantially smaller, and the problem is unlikely to be complete in $\theta2$.

One typical $\theta2$-complete problem is \vcmember due to~\citet{hemaspaandra2005complexity}.

\begin{problem}[\vcmember]\label{problem:vcmember}
    Given an undirected graph $G=(V,E)$ and a vertex $x\in V$, decide if there is a minimum-size vertex cover $S$ that contains $x$.
\end{problem}

To solve \vcmember, we can make the following two queries for each $k=1,\ldots,|V|$:
\begin{enumerate}
    \item Does there exist a vertex cover of size strictly smaller than $k$;
    \item Does there exist a vertex cover of size exactly $k$ that contains $x$.
\end{enumerate}
Both queries are obviously NP queries.
It is easy to see that the \vcmember instance is a \yes instance if and only if the answers to the above two queries are \no and \yes for some value $k$.
In addition, letting $k^\ast$ be the size of the minimum vertex cover, the answers to both queries are \no for each $k<k^\ast$, and the answers to both queries are \yes for each $k>k^\ast$.
We need to exactly locate $k^\ast$ to figure out if the \vcmember instance is a \yes instance.
Since $k^\ast$ is not known in advance and the minimum vertex cover is a classical $\classNP$-hard optimization problem, we have to non-adaptively make queries for all the values of $k$, or adaptively make logarithmically many queries by a binary search.

For our fair division problems, with general monotone valuation functions, both \efpo and \efmsw are in $\theta2$ for $n$ and $k$ being both constants.
Consider the value truth table $\mathcal{T}$.
It has a polynomial number of cells, and the truth-value to each cell can be figured out by an NP query.
As we have seen earlier, once the table is constructed, it is routine to check if there is an allocation that is envy-free and efficient (social welfare optimal for \efmsw, Pareto-optimal for \efpo).
On the other hand, we can easily find some $\classNP$-complete problem to reduce from, and show that obtaining the truth-value of each cell is $\classNP$-hard.
Thus, different from the scenarios with additive valuations (where the value of each cell can be solved by dynamic programming), we have to make a query to obtain the value of each cell.

The problem \efmsw-$[\infty,\mono,k]$ is also in $\theta2$ for any constant $k$.
Although the table $\mathcal{T}$ no longer has a polynomial size, the problem is still relatively simpler for \efmsw, compared with \efpo.
For \efmsw, we do not have to construct $\mathcal{T}$.
Instead, for each possible social welfare $t$, we can make the following two NP queries:
\begin{enumerate}
    \item Does there exist an allocation with social welfare at least $t$;
    \item Does there exist an envy-free allocation with social welfare at least $t$.
\end{enumerate}
Notice that a \yes answer to the second query implies a \yes answer to the first, and we get a \no instance of \efmsw if and only if the answers to the two queries are different for some value $t$ (i.e., \yes to the first question and \no to the second question).
Since the maximum possible social welfare is $nmk$ for $k$-ary valuations, we need only a polynomial number of non-adaptive queries, or a logarithmic number of adaptive queries with a binary search.

On the other hand, all the fair division problems mentioned above are complete in $\theta2$.
We use \efmsw-$[2,\mono,2]$ to demonstrate the $\theta2$-completeness.
When solving the problem by NP oracles, we face the same challenges as they are in \vcmember: we do not know what is the optimal social welfare, and we have to search over possible values of optimal social welfare.
On the other hand, if the optimal social welfare is known, or if the optimal social welfare is known to be in a small range, the problem can no longer be $\theta2$-complete, as fewer NP queries can solve the problem.
Therefore, when we are doing the reduction, the hard \efmsw-$[2,\mono,2]$ instances must be those where the optimal social welfare depends on an $\classNP$-complete problem and can be any value in a \emph{wide range}.
This implies that the value for each of the two agents must be able to vary over a wide range: if agent $1$'s value is bounded by some constant $c$, we know that the optimal social welfare is in the narrow range $[v_2(M),v_2(M)+c]$ (in this case we only need $2(c+1)$ queries, two for each possible social welfare value).
As we will see in Sect.~\ref{sect:hardness}, such a reduction exists, and the reduction is from \vcmember.
All the other $\theta2$-complete fair division problems are proved by reductions from \vcmember. 

\subsection{Complexity Class $\delta2$ and Fair Division Problems Contained in}
As mentioned earlier, problems in $\delta2$ are precisely the ones that can be solved by making a polynomial number of adaptive queries to an NP oracle.
By the similar discussions at the beginning of the previous sub-section, for a problem complete in $\delta2$, the numbers of ``accepting paths'' and ``rejecting paths'' in the decision tree must both be large.
If one of them is polynomial, then the problem can be solved by making a polynomial number of non-adaptive queries, and it is in $\theta2$.
The following problem is known to be $\delta2$-complete due to~\citep{krentel1986complexity}.

\begin{problem}[\maxoddSAT] \label{problem:maxoddSAT}
Given a Boolean formula $\phi$ with variables $x_1,\ldots,x_n$, determine if $\phi$ is satisfiable and its lexicographically last satisfying assignment has $x_n=1$.
\end{problem}

For a \yes instance of \maxoddSAT, $\phi$ has to be satisfiable, and the ``maximum'' solution must be an odd number.
This problem can be solved with $2n$ adaptive queries (for $n$ variables), as we can iteratively and tentatively assign $1$ or $0$ for each $x_i$ to see if the resultant formula is still satisfiable.
Exactly half of the paths in the decision tree lead to \yes.
This is why the problem is hard in $\delta2$.
In fact, many other problems regarding the parity of optimal solutions are $\delta2$-hard.
It should be noted that the range of the optimal solution must be in $[0,T]$ for some $T$ with a polynomial \emph{length in binary representation}, not polynomial \emph{value}.
If the value of $T$ is polynomial, the problem can be solved in $O(T)$ queries, and the problem is in $\theta2$.
For example, checking the parity of the optimal solution to the Traveling Salesman Problem is $\delta2$-complete, while checking the parity of the optimal maximum clique solution is in $\theta2$.

For fair division problems, all problems in \efmsw are in $\delta2$.
Similar to the discussions in the previous sub-section, we can make two queries for each possible social welfare value $t$.
For the problems in \efmsw-$[\cdot,\mono,\infty]$ where there are no $k$-ary assumptions, the upper bound of optimal social welfare $t$ is an exponentially large number with a polynomial length in binary representation.
We have to use a polynomial number of adaptive queries with a binary search.

\subsection{Fair Division Problems in $\sigma2$}
For the remaining fair division problems in Table~\ref{tab:results}, i.e., those \efpo problems that are not discussed in the previous sub-sections, there is no clear way to solve them by an NP oracle in polynomial time.
In addition, there is no clear evidence that they are also in $\pi2$.
As it turns out, most of them are $\sigma2$-complete.
Below, we list the $\sigma2$-complete problems we used for reduction.

The first problem we use is the complement of the well-known $\pi2$-complete problem \aeCNF, which is $\sigma2$-complete.

\begin{problem}[Complement of \aeCNF]\label{problem:aecnf}
        Given a \textsc{3CNF} formula $\phi$ with variables partitioned into $V_{\forall}$ and $V_{\exists}$, determine if there exists a Boolean assignment to the variables in $V_\forall$ such that $\phi$ evaluates to ``false'' for all Boolean assignments of $V_\exists$.
\end{problem}

Another $\sigma2$-complete problem is the following one due to~\citet{berman2002complexity}.

\begin{problem}[\easubsetsum]\label{problem:easubsetsum}
    Given two multi-sets of positive integers $V_\exists,V_\forall$ and an integer $T$, decide:
    \begin{itemize}
        \item[\yes:] there exists $U_\exists\subseteq V_\exists$ such that for all $U_\forall\subseteq V_\forall$ the sum of all integers in $U_\exists\cup U_\forall$ is \emph{not} $T$;
        \item[\no:] for all $U_\exists\subseteq V_\exists$ there exists $U_\forall\subseteq V_\forall$ such that the sum of all integers in $U_\exists\cup U_\forall$ is $T$.
    \end{itemize}
\end{problem}

Intuitively, \aeCNF is handy when we are showing the $\sigma2$-completeness for problems with general monotone valuations.
We can construct two items $x_i$ and $\neg x_i$ for each variable.
We can let an agent take exactly one from each group $\{x_i,\neg x_i\}$ to represent a value assignment (e.g., we can set the valuation of the agent to simply $0$ if none of $\{x_i,\neg x_i\}$ is included in the bundle). We can check if there is at least one item/literal chosen in each clause and define the value to $0$ if not. This simulates satisfiability.
Some other gadgets with different purposes are constructed to make the reduction work.

For additive valuations with a constant number of agents, the problem \easubsetsum is more useful.
For a fixed number of agents and polynomially bounded integer values, the valuation truth table can be constructed in polynomial time. 
Therefore, a hardness reduction for this setting must avoid polynomially bounded utilities, unless the corresponding complexity classes collapse.
This makes \easubsetsum more suitable. On the other hand, the reduction used by~\citet{de2009complexity} constructs agents' valuations that are bounded by polynomials of the input, which does not work here.

When considering $k$-ary additive valuations, the construction in~\citet{de2009complexity} does not work as the valuations are polynomial.
The above two $\sigma2$-complete problems are also hard to adapt here, as we may need to construct values equal to the number of variables/clauses in the \aeCNF instance, or the integers in the \easubsetsum instance, which are not constant.
Instead, in Sect.~\ref{sec:infty0k-po}, we introduce another $\sigma2$-complete problem, \neeakrs, which is a $\sigma2$-variant of the well-known NP-complete \krs problem.

\begin{problem}[\neeakrs]\label{problem:neeakrs}
    Given a bipartite graph $G=(U\cup V,E)$ with $U$ partitioned into $U_{\exists}$ (with $|U_\exists|\geq2$) and $U_{\forall}$ and a positive integer $\kappa\ge3$, decide:
    \begin{itemize}
        \item[\yes:] there exists a non-empty subset $S_{\exists}\subseteq U_{\exists}$, such that for any $S_{\forall}\subseteq U_{\forall}$, $G$ does not contain any $\kappa$-regular subgraph $G'=(U'\cup V',E')$ with $U'=S_{\exists}\cup S_{\forall}$.
        
        \item[\no:] for any non-empty subset $S_{\exists}\subseteq U_{\exists}$, there exists some $S_{\forall}\subseteq U_{\forall}$ such that $G$ contains a non-empty $\kappa$-regular subgraph $G'=(U'\cup V',E')$ with $U'=S_{\exists}\cup S_{\forall}$.
    \end{itemize}
\end{problem}

Intuitively, in our proof based on the bipartite graph, we construct additional vertices and edges to get a new graph $\mathcal{G}$.
Each vertex represents an agent and each edge represents an item.
Each vertex agent has a positive value to its adjacent edge item, which is bounded by some constant $k$ (which is related to $\kappa$).
Under a \no instance of \neeakrs, for any envy-free allocation, we may find a corresponding $\kappa$-regular subgraph, which results in a ``Pareto-improvement cycle'' in $\mathcal{G}$ by ``pushing items along the cycle''.
However, the resulting allocation is no longer envy-free.
On the other hand, under a \yes instance, there exists an envy-free allocation with no corresponding $\kappa$-regular subgraph, which prevents the undesirable Pareto-improvement.

\section{\efpo under $k$-ary Additive Valuations}\label{sec:infty0k-po}
In this section, we show the $\sigma2$-hardness of \efpo-$[\infty,\add,3]$.
In Sect.~\ref{sec:kregular}, we present the $\sigma2$-complete variants of the \krs problem.
Then, we give our main proof in Sect.~\ref{sec:infty0k-po-sub}, where we first show a variant of \efpo with additional constraints on each agent's received utility (\boundefpo) is $\sigma2$-hard, and then reduce it to \efpo-$[\infty,\add,3]$.
This solves the open problem proposed by~\citet{bouveret2008efficiency}.

\subsection{$\sigma2$-Completeness of \krs Variants}\label{sec:kregular}
The following two problems, Problem~\ref{problem:eakrs} and the nonempty variant defined in Problem~\ref{problem:neeakrs}, are $\sigma2$-complete.
Compared with Problem~\ref{problem:eakrs}, Problem~\ref{problem:neeakrs} requires the existentially quantified subset to be non-empty for a \yes instance.

\begin{problem}[\eakrs]\label{problem:eakrs}
    Given a bipartite graph $G=(U\cup V,E)$ with $U$ partitioned into $U_{\exists}$ and $U_{\forall}$ and a positive integer $\kappa\ge3$, decide:
    \begin{itemize}
        \item[\yes:] there exists $S_{\exists}\subseteq U_{\exists}$, such that for any $S_{\forall}\subseteq U_{\forall}$, $G$ does not contain any non-empty $\kappa$-regular subgraph $G'=(U'\cup V',E')$ with $U'=S_{\exists}\cup S_{\forall}$.
        
        \item[\no:] for any $S_{\exists}\subseteq U_{\exists}$, there exists $S_{\forall}\subseteq U_{\forall}$ such that $G$ contains a non-empty $\kappa$-regular subgraph $G'=(U'\cup V',E')$ with $U'=S_{\exists}\cup S_{\forall}$.
    \end{itemize}
\end{problem}

\begin{proposition}\label{prop:kregular-sigma2}
    For $\kappa=3$, the problem \eakrs is $\sigma2$-complete, even if the bipartite graph has maximum degree $\kappa+1$ and contains no multiple edges.
\end{proposition}
\begin{proof}
    Fix $\kappa=3$ in this proof.
    The problem clearly belongs to $\sigma2$, as the existential certificate can be a subset $S_{\exists}\subseteq U_{\exists}$, and the universal certificate can be the $\kappa$-regular subgraph described in the problem.

    To show the $\sigma2$-hardness, we provide a reduction from the known $\sigma2$-complete problem \eaxc~\citep{mcloughlin1984complexity}.
    \begin{problem}[\eaxc]
        Given a set $X$ of $3q$ elements and two disjoint collections $C_1, C_2$ of $3$-element subsets of $X$. decide if there exists a subcollection $C'_1\subseteq C_1$ such that for any subcollection $C'_2\subseteq C_2$, the union $C'_1\cup C'_2$ is not an exact cover of $X$.
        In an exact cover, each element in $X$ appears exactly once.
    \end{problem}
    We follow the construction given by~\citet{plesnik1984note}, who shows that deciding whether a bipartite graph $G=(U\cup V,E)$ contains a $\kappa$-regular subgraph is $\classNP$-complete by reducing from \xc (given a set $X$ of $3q$ elements and a collection $C$ of the subsets, decide if there exists a subcollection of $C$ that is an exact cover) with groundset $X$ and a collection of subsets $C$.
    Note that the resulting bipartite graph has maximum degree $\kappa+1$ and contains no multiple edges.
    We omit the details here. 
    It suffices to know that in their construction, each subset in $C$ corresponds to a vertex in $U$ with degree $\kappa$, and each element in $X$ corresponds to a vertex in $V$.
    There are additional structures (including extra vertices) to guarantee that any $\kappa$-regular subgraph will contain all vertices corresponding to elements in $X$.
    For each exact cover $C'\subseteq C$ of $X$, there is a unique $\kappa$-regular subgraph that contains all vertices corresponding to $C'$ and contains no vertices corresponding to $C\setminus C'$.
    Conversely, for each $\kappa$-regular subgraph, the vertices contained in the subgraph that corresponds to $C$ form an exact cover.
    
    In our problem, given an \eaxc instance with set $X$ and two collections $C_1$ and $C_2$, we apply the above construction with $C=C_1\cup C_2$ to obtain a bipartite graph $G=(U\cup V,E)$.
    Let $U_{\exists}\subseteq U$ be the set of vertices corresponding to $C_1$, and let $U_{\forall}=U\setminus U_{\exists}$.
    Assume the \eaxc instance is a \no instance.
    Consider an arbitrary subset $U'_{\exists}\subseteq U_{\exists}$ in the \eakrs instance.
    It corresponds to a subcollection $C'_1\subseteq C_1$, and we can find $C'_2\subseteq C_2$ such that $C'_1\cup C'_2$ is an exact cover of $X$.
    The exact cover will correspond to a $\kappa$-regular subgraph in $G$ by the analysis above, where no vertex in $U_\exists\setminus U'_{\exists}$ is included.
    On the other hand, assume the \eaxc instance is a \yes instance, and let $C'_1\subseteq C_1$ be the subcollection that does not extend to an exact cover.
    Let $U'_{\exists}\subseteq U_{\exists}$ be the corresponding subset.
    There is no $\kappa$-regular subgraph that contains all vertices in $U'_{\exists}$ and no vertices in $U_{\exists}\setminus U'_{\exists}$, for otherwise it indicates the existence of an exact cover with $C'_1$.
    Therefore, we conclude the $\sigma2$-hardness of \eakrs.
\end{proof}

We now consider the nonempty variant defined in Problem~\ref{problem:neeakrs}.

\begin{proposition}\label{prop:nekregular-sigma2}
    For $\kappa=3$, the problem \neeakrs is $\sigma2$-complete, even if the bipartite graph has maximum degree $\kappa+1$ and contains no multiple edges.
\end{proposition}
\begin{proof}
    The problem clearly belongs to $\sigma2$.
    To show the $\sigma2$-hardness, we reduce from the complement of \aeCNF based on the construction by~\citet{mcloughlin1984complexity} and~\citet{plesnik1984note}.

    \citet{mcloughlin1984complexity} shows the $\sigma2$-completeness of \eaxc by reducing from the complement of \aeCNF with formula $\phi$ with variables partitioned into $V_{\forall}$ and $V_{\exists}$, denoted by $(\phi, V_{\forall}, V_{\exists})$.
    In particular, for each variable in the $3$CNF formula, several truth-value-setting sets are constructed.
    The first collection $C_1$ in the \eaxc instance contains one truth-value-setting set for each universally quantified variable in $V_{\forall}$, while the second collection $C_2$ contains all other sets.
    The insight is that, any exact cover will contain either all or none truth-value-setting sets for each variable.
    A satisfying assignment will correspond to an exact cover by including all truth-value-setting sets for the variables assigned true, and excluding all truth-value-setting sets for those assigned false.
    Conversely, a satisfying assignment can be constructed from an exact cover by setting the variable to true if its truth-value-setting sets are included in the exact cover, and false otherwise.

    We present our construction for the \neeakrs problem as follows.
    Given an instance $(\phi, V_{\forall}, V_{\exists})$ of the complement of \aeCNF, we add a new variable $x$ to $V_{\forall}$ that does not appear in $\phi$.
    We also add a new clause $(x\vee x\vee \neg x)$ to $\phi$.
    Let the new instance be $(\phi^\dag,V_{\exists}^\dag,V_{\forall}^{\dag})$.
    Note that $\phi$ and $\phi^\dag$ have the same satisfiability (we can either set $x$ to true or false without affecting the satisfiability).
    We first apply the construction by~\citet{mcloughlin1984complexity} on $(\phi^\dag,V_{\exists}^\dag,V_{\forall}^{\dag})$ to obtain an \eaxc instance $(X, C_1, C_2)$, then apply the construction in Proposition~\ref{prop:kregular-sigma2} on $(X, C_1, C_2)$ to obtain a \neeakrs instance $\left(G=(U\cup V, E), U_{\exists}, U_{\forall}\right)$.
    In particular, the bipartite graph $G$ has maximum degree $\kappa+1$ with no multiple edges.

    The constructed \neeakrs instance is a \yes instance if and only if the complement of the \aeCNF instance is a \yes instance.
    Assume $(\phi,V_{\forall},V_{\exists})$ is a \yes instance, that is, there exists an assignment to $V_{\forall}$ such that $\phi$ cannot be satisfied under any assignment to $V_{\exists}$.
    The same assignment for $V_{\forall}^\dag$ with $x$ set to true will ensure that $\phi^\dag$ cannot be satisfied under any assignment to $V_{\exists}^\dag$.
    By the analysis of~\citet{mcloughlin1984complexity}, there exists a non-empty subcollection $C'_1\subseteq C_1$ to prevent any exact cover in $(X,C_1,C_2)$, where $C'_1$ contains only the truth-value-setting sets for variables in $V_{\forall}$ that are assigned true (hence $C'_1$ will contain the truth-value-setting sets for $x$).
    Consequently, by the analysis of Proposition~\ref{prop:kregular-sigma2}, there exists a non-empty subset in $U_{\exists}$ (contains the vertex corresponding to the truth-value-setting set for $x$) to prevent any $\kappa$-regular subgraph in $G$, thus the \neeakrs instance is a \yes instance.

    On the other hand, assume $(\phi,V_{\forall},V_{\exists})$ is a \no instance.
    Then, any assignment to $V_{\forall}^\dag$ can be extended to a satisfying assignment to $\phi^\dag$ as $(x\vee x\vee\neg x)$ is always satisfiable.
    According to the above analysis, the \eaxc instance is a \no instance, thus the \neeakrs instance is also a \no instance.
    This completes the proof.
\end{proof}

\subsection{$\sigma2$-hardness of \efpo-$[\infty,\add,3]$}\label{sec:infty0k-po-sub}
In this part, we first show the $\sigma2$-hardness of an intermediate problem called \emph{\boundefpo}.
The problem takes as input a \emph{bounded} fair division instance $(N,M,\{v_i\}_{i\in N}, \{(u_i^+, u_i^-)\}_{i\in N})$.
Compared to the classical fair division instance, each agent $i\in N$ is associated with two non-negative integers $u_i^+$ and $u_i^-$, representing the upper and lower bounds of her target utility received from the allocation.
An allocation $(A_1,\ldots,A_n)$ is said to be \emph{feasible to agent $i$} if $u_i^-\le v_i(A_i)\le u_i^+$, and is said to be \emph{feasible} if it is feasible to all agents.
\begin{problem}
    For a bounded fair division instance $(N,M,\{v_i\}_{i\in N}, \{(u_i^+, u_i^-)\}_{i\in N})$, decide:
    \begin{itemize}
        \item[\yes:] there exists a feasible, envy-free and Pareto-optimal allocation, where feasibility indicates that $u_i^-\le v_i(A_i)\le u_i^+$ for every $i\in N$;
        \item[\no:] every Pareto-optimal allocation is either not envy-free or infeasible.
    \end{itemize}
\end{problem}
The original \efpo problem can be viewed as a special case of \boundefpo with $u_i^+=\infty$ and $u_i^-=0$ for each agent $i\in N$.

We first show the problem is $\sigma2$-hard for $4$-ary valuations.
The proof for $3$-ary valuations is based on the following construction.
\begin{lemma}\label{lem:hard:infty04-bpo}
    \boundefpo-$[\infty,\add,4]$ is $\sigma2$-hard.
\end{lemma}
\begin{proof}
    We present a reduction from \neeakrs, which is shown to be $\sigma2$-complete in Proposition~\ref{prop:nekregular-sigma2}.
    We focus on the case with $\kappa=3$.
    Recall that for a given \neeakrs instance $$\left(G=(U\cup V,E), U_{\exists}, U_{\forall}\right)$$ with $\kappa=3$, it is a \yes instance if and only if there exists a non-empty subset $S_\exists\subseteq U_\exists$, such that for any $S_\forall\subseteq U_\forall$, $G$ does not contain any $3$-regular subgraph $G'=(S_\exists\cup S_\forall\cup V', E')$ for any $V'\subseteq V$ and $E'\subseteq E$.
    
    We first construct a bounded fair division instance $\mathcal{I}=(N,M,\{v_i\}_{i\in N}, \{(u_i^+, u_i^-)\}_{i\in N})$ (note that at this stage, the construction is still incomplete and does not finish the proof.
    We will establish some useful insights on $\mathcal{I}$).
    Assume that $|U|=|V|=t$ (by adding isolated vertices to the smaller side of the bipartition; any added left-side vertices should belong to $U_\forall$), and let $t_1=|U_{\exists}|$ and $t_2=|U_{\forall}|$.
    Notice that $t_1>1$ by the definition of the problem.
    
    We first describe the construction of the set of agents, as well as the bounds on their target utilities $u_i^+$ and $u_i^-$.
    \begin{itemize}
        \item A set of $t=t_1+t_2$ agents $N_a$ that corresponds to $U$. This set $N_a$ is further partitioned into a set of $t_1$ agents $N_\exists=\{a_1^\exists,\ldots,a_{t_1}^{\exists}\}$, each corresponds to a vertex in $U_\exists$, and a set of $t_2$ agents $N_{\forall}=\{a_1^\forall,\ldots,a_{t_2}^{\forall}\}$, each corresponds to a vertex in $U_\forall$.
        For agent $i \in N_a$, assume its corresponding vertex in $G$ has degree $d_i$, then we set $u_i^+ = u_i^- = d_i$.
        In the following, when we refer to an agent in $N_a$ without distinguishing $N_\exists$ and $N_\forall$, the superscripts in $a_i^\exists$ and $a_i^\forall$ may be omitted and we simply write $a_i$.
        
        \item A set of $t$ agents $N_b=\{b_1,\ldots,b_t\}$, each corresponds to a vertex in $V$.
        Let $u_i^+ = u_i^- = 3$ for each agent $i\in N_b$.

        \item A set of $t_1$ agents $N_{c^1}=\{c_1^1,\ldots,c_{t_1}^1\}$.
        Let $u_i^+ = u_i^- = 3$ for each agent $i\in N_{c^1}$.

        \item A set of $t_1$ agents $N_{c^2}=\{c_1^2,\ldots,c_{t_1}^2\}$.
        Let $u_i^+ =\infty$ and $u_i^- = 1$ for each agent $i\in N_{c^2}$.

        \item A set of $t_1$ agents $N_{c^3}=\{c_1^3,\ldots,c_{t_1}^3\}$.
        Let $u_i^+ =\infty$ and $u_i^- = 0$ for each agent $i\in N_{c^3}$.

        \item Two additional agents $z_1$ and $z_2$.
        Let $u_{z_1}^+ = u_{z_1}^- = 3t$, and $u_{z_2}^+ = u_{z_2}^- = 3t_1$.
    \end{itemize}
    
    We now present the construction of the set of items, along with the agents who positively value them.
    \begin{itemize}
        \item Each edge $e_i=(u_j,v_{j'})\in E$ corresponds to an item $g^e_i$.
        Both agent $a_j$ corresponding to $u_j$ and agent $b_{j'}$ corresponding to $v_{j'}$ have value $1$ to $g^e_i$.
        The set of all such items constructed from edges is denoted by $M_e$.

        \item A set of $t_1$ items $M_{a}=\{g_1^a,\ldots,g_{t_1}^a\}$.
        Each agent $a_i^\exists\in N_{\exists}$ has value $3$ to item $g_i^a$, and each agent $c_i^2\in N_{c^2}$ has value $2$ to item $g_i^a$.
        
        \item A set of $t$ items $M_b=\{g_1^b,\ldots,g_t^b\}$.
        Each agent $b_i\in N_b$ has value $3$ to item $g_i^b$, and agent $z_1$ has value $3$ to each item in $M_b$.
        
        \item A set of $t_1$ items $M_{\exists}=\{g_1^{\exists},\ldots,g_{t_1}^{\exists}\}$.
        Each agent $c_i^1\in N_{c^1}$ has value $3$ to item $g_i^{\exists}$, and agent $z_1$ has value $3$ to each item in $M_{\exists}$.

        \item A set of $t_2$ items $M_{\forall}=\{g_1^{\forall},\ldots,g_{t_2}^{\forall}\}$.
        Each agent $a_i^\forall\in N_{\forall}$ has value $3$ to item $g_i^{\forall}$, and agent $z_1$ has value $3$ to each item in $M_{\forall}$.

        \item A set of $t_1$ items $M_{c^1}=\{g_1^1,\ldots,g_{t_1}^1\}$.
        Each agent $c_i^1\in N_{c^1}$ has value $3$ to item $g_i^1$, and agent $z_2$ has value $3$ to each item in $M_{c^1}$.

        \item A set of $t_1$ items $M_{c^2}=\{g_1^2,\ldots,g_{t_1}^2\}$.
        Each agent $c_i^2\in N_{c^2}$ has value $3$ to item $g_i^2$, and agent $z_2$ has value $3$ to each item in $M_{c^2}$.

        \item A set of $t_1$ items $M_{c^3}=\{g_1^3,\ldots,g_{t_1}^3\}$.
        Each agent $c_i^3\in N_{c^3}$ has value $3$ to item $g_i^3$, and agent $z_2$ has value $3$ to each item in $M_{c^3}$.
    \end{itemize}
    An agent's value to an item is $0$ if it is not explicitly defined above.

    An illustration of the main construction is given in Fig.~\ref{fig:efpo}.
    Noticing that each item is positively valued by exactly two agents, the edges in the figure correspond to items (ignore the directions of the edges at this moment).
    Agents in $N_{c^3}$ and items in $M_{c^3}$ are not shown in the figure for conciseness.
    We say that an allocation is \emph{non-wasteful} if each item/edge is only allocated to one of its endpoints.
    Clearly, a Pareto-optimal allocation must be non-wasteful, and we will focus only on non-wasteful allocations from now on.

    As we have mentioned at the beginning, the construction of the instance is not fully completed yet.
    An additional set of agents $N_f$ and an additional set of items $M_f$ will be further constructed at the end, after we have explored some useful properties of the current incomplete instance.
    
    \begin{figure}[h]
    \centering

    \resizebox{0.85\textwidth}{!}
    {
    \begin{tikzpicture}[
        node distance=1.5cm and 3cm,
        g_agent/.style={rectangle, draw=black!90, rounded corners, minimum height=0.8cm, font=\small, align=center, font=\bfseries},
        agent/.style={rectangle, draw=black!90, fill=white, rounded corners, minimum size=0.8cm, font=\small, align=center, font=\bfseries},
        item_transferred/.style={->, >={Stealth[length=3mm]},  thick, color=red!80!black,rounded corners=4pt},
        item_transfer/.style={->, >={Stealth[length=3mm]},  thick, color=black,rounded corners=4pt},
        edge/.style={-, thick, color=gray!50!black,rounded corners=4pt},
        label_text/.style={font=\textbf\itshape, align=center, color=black},
        section_box/.style={draw, dashed, rounded corners, inner sep=8pt},
    ]

    \node[agent] (z2) at (-0.5, 3.75) {$z_2$};

    \node[g_agent, minimum width=5.5cm] (nc2) at (0, 2) {Agents $N_{c^2}$};

    \node[g_agent, align=right, minimum width=5.5cm,label=left:{Set $U_\exists$}] (nexists) at (0, -1) {$H_\exists\qquad\quad I_\exists\qquad\qquad\qquad\qquad$};

    \node[] at (0, -2) {\bfseries Agents $N_\exists$};

    \node[g_agent, minimum width=2.5cm,label=right:{Set $U_\forall$}] (nforall) at (5, -1) {Agents $N_\forall$};

    \node[g_agent, minimum width=9cm, label=right:{Set $V$}] (nb) at (1.75, -3.5) {Agents $N_b$};

    \node[g_agent, minimum width=3cm, rotate=0] (nc1) at (-5.2, 2) {Agents $N_{c^1}$};

    \node[agent] (z1) at (-8, 0) {$z_1$};

    \draw[item_transferred] (z1.north) -- node[near start, right, label_text] {$M_\forall$} ++(0,4.6) -- ++(12.23,0) -- ($(nforall.north west)!0.2!(nforall.north east)$);

    \node[color=red!70!black] at (-3.7, 4.4) {\textbf{Pareto-improvement cycle}};

    \draw[item_transferred] ($(nforall.south west)!0.2!(nforall.south east)$) -- ($(nb.north west)!0.777!(nb.north east)$);

    \draw[item_transferred] ($(nc1.north west)!0.2!(nc1.north east)$) |- ($(z2.north west)!0.2!(z2.south west)$) node[near start, left, label_text] {$M_{c^1}$};
    \draw[edge] ($(nc1.north west)!0.4!(nc1.north east)$) |- ($(z2.north west)!0.4!(z2.south west)$);
    \draw[edge] ($(nc1.north west)!0.6!(nc1.north east)$) |- ($(z2.north west)!0.6!(z2.south west)$);

    \draw[item_transferred] ($(z1.north east)!0.2!(z1.south east)$) -| ($(nc1.south west)!0.2!(nc1.south east)$) node[near start, above, label_text] {$M_\exists$};
    \draw[edge] ($(z1.north east)!0.4!(z1.south east)$) -| ($(nc1.south west)!0.4!(nc1.south east)$);
    \draw[edge] ($(z1.north east)!0.6!(z1.south east)$) -| ($(nc1.south west)!0.6!(nc1.south east)$);

    \draw[item_transferred]  ($(nb.south west)!0.777!(nb.south east)$) -- ++(0,-0.8) -- ++(-12.4,0) -- ++(0,4.3);
    \draw[item_transferred]  ($(nb.south west)!0.324!(nb.south east)$) -- ++(0,-0.4) -- ++(-8.1,0) -- ++(0,3.9);
    \draw[edge] (nb.west) -| ($(z1.south west)!0.78!(z1.south east)$) node[near start, above, label_text] {$M_b$};

    \draw[item_transfer] (z2) -- ($(nc2.north west)!0.1!(nc2.north east)$) node[near start, left, label_text] {$M_{c^2}$};
    \draw[item_transfer] ($(nc2.south west)!0.1!(nc2.south east)$) -- ($(nexists.north west)!0.1!(nexists.north east)$) node[near end, left, label_text] {$M_a$};
    \draw[item_transfer] ($(nexists.south west)!0.1!(nexists.south east)$) -- ($(nb.north west)!0.062!(nb.north east)$) node[midway, left, label_text] {$M_e$};

    \draw[item_transfer] ($(nc2.north west)!0.35!(nc2.north east)$) -- (z2);
    \draw[item_transfer] ($(nexists.north west)!0.35!(nexists.north east)$) -- ($(nc2.south west)!0.35!(nc2.south east)$);

    \draw[item_transfer] (z2) -- ($(nc2.north west)!0.88!(nc2.north east)$);
    \draw[item_transfer] ($(nexists.north west)!0.88!(nexists.north east)$) -- ($(nc2.south west)!0.88!(nc2.south east)$);

    \draw[item_transferred] (z2) -- ($(nc2.north west)!0.53!(nc2.north east)$);
    \draw[item_transferred] ($(nc2.south west)!0.53!(nc2.south east)$) -- ($(nexists.north west)!0.53!(nexists.north east)$);
    \draw[item_transferred] ($(nexists.south west)!0.53!(nexists.south east)$) -- ($(nb.north west)!0.324!(nb.north east)$);

    \node[align=center,font=\itshape,fill=white,inner sep=-0pt] at (-2.2, 0.8) {Type-I\\Agents};
    \node[align=center,font=\itshape,fill=white,inner sep=0.1pt] at (-0.9, 0.8) {Type-II\\Agents};
    \node[align=center,font=\itshape,fill=white,inner sep=-0.pt] at (2.05, 0.8) {Type-III\\Agents};

    \begin{scope}[on background layer]
        \node[section_box, fit=($(nc2.north west)!0.04!(nc2.north east)$)($(nexists.south west)!0.16!(nexists.south east)$), label={[shift={(2,-0.5)}]}] {};

        \node[section_box, fill=red!5, fit=($(nforall.north west)!0.1!(nforall.north east)$)($(nb.south west)!0.8!(nb.south east)$), label={[shift={(2,-0.5)}]mid:\small \textcolor{red!70!black}{\text{$3$-regular subgraph}}}] {};

        \node[section_box, draw=none, fill=red!5, fit=($(nb.south west)!0.32!(nb.south east)$)($(nc2.north west)!0.62!(nc2.north east)$), label={[shift={(2,-0.5)}]}] {};

        \node[section_box, fit=($(nexists.south west)!0.67!(nexists.south east)$)($(nc2.north west)!0.28!(nc2.north east)$), label={[shift={(2,-0.5)}]}] {};

        \node[section_box,fit=($(nexists.south west)!0.96!(nexists.south east)$)($(nc2.north west)!0.8!(nc2.north east)$), label={[shift={(2,-0.5)}]}] {};
    \end{scope}

    \end{tikzpicture}
    }
    \caption{Visualization of main construction in Lemma~\ref{lem:hard:infty04-bpo}. Vertices represent agents and edges represent items. 
    At this stage, each item is positively valued by two agents, and agents positively value the adjacent items. 
    The graph shows a partial feasible allocation induced by some feasible tuple, where a directed black edge from a source agent to a sink agent represents that the edge item is allocated to the sink agent in the allocation. 
    In a \no instance, a Pareto-improvement cycle of the feasible allocation is formed by the directed red edges, which corresponds to a $3$-regular subgraph of $G$. 
    The cycle includes some of the Type-II Agents. 
    The utility of each agent $c_i^2$ in $N_{c^2}$ on the cycle strictly increases by $1$ by receiving item $g_i^2$ in $M_{c^2}$ with value $3$, and giving away item $g_i^a$ with value $2$ in $M_a$ to agent $a_i^\exists$ in $N_\exists$. 
    The bundle for each agent in $N_c^2$ belonging to Type-I and Type-III will not change throughout the Pareto-improvement. In a \yes instance, there is no such Pareto-improvement cycle for some feasible allocations.}
    \label{fig:efpo}
    \end{figure}

    Intuitively, to guarantee $z_2$'s feasibility, she needs to receive at least $t_1$ items from $M_{c^1}\cup M_{c^2}\cup M_{c^3}$.
    We will show that, in a \no instance of the \neeakrs problem, in any feasible and non-wasteful allocation, if $z_2$ receives some items from $M_{c^2}$, there exists a Pareto-improvement where $z_2$ gives away all items taken from $M_{c^2}$ to agents in $N_{c^2}$ and receives items only from $M_{c^1}\cup M_{c^3}$.
    We will then construct a set of agents $N_f$ that has value $1$ to each item in $M_{c^1}\cup M_{c^3}$ and each receives a value of at most $t_1-1$ in any envy-free allocation (the detailed construction will be shown later).
    Therefore, if $z_2$ receives no less than $t_1$ items only from $M_{c^1}\cup M_{c^3}$, the allocation is no longer envy-free to agents in $N_f$.
    On the other hand, in a \yes instance, $z_2$ can receive items from $M_{c^2}$ in a Pareto-optimal allocation, thus guaranteeing envy-freeness from the $N_f$ to $z_2$.

    We first define a collection of allocations which are called ``induced from a feasible tuple''.
    We will show every allocation that cannot be ``induced from a feasible tuple'' will violate either feasibility or non-wastefulness.

    For simplicity, we say that $H=(H_a\cup H_b\subseteq N_a\cup N_b, O_e\subseteq M_e)$ is a subgraph of $G$, if the corresponding vertices and edges from which the agents in $H_a\cup H_b$ and items in $O_e$ are constructed, form a subgraph of $G$.
    We further assume that $H_a=H_\exists\cup H_\forall$, where $H_\exists\subseteq N_\exists$ and $H_\forall\subseteq N_\forall$.

    Consider the following tuple $$\left(H=(H_\exists\cup H_\forall\cup H_b\subseteq N_\exists\cup N_\forall\cup N_b, O_e\subseteq M_e), I_\exists\subseteq N_{\exists}, I_{c^1}\subseteq N_{c^1}, I_{c^3}\subseteq N_{c^3}\right).$$
    We say the tuple is \emph{feasible} if it satisfies: $H$ is a $3$-regular bipartite subgraph of $G$ that can possibly be empty, $H_\exists\cap I_\exists=\emptyset$, $|H_\exists|=|I_{c^1}|$, and $|I_{c^3}|=t_1-|H_\exists\cup I_\exists|$.

    A feasible tuple partitions agents in $N_\exists\cup N_{c^2}$ into three types, which is shown in Fig.~\ref{fig:efpo}.
    \begin{itemize}
        \item Type-I Agents contain agents $a_i^\exists\in H_\exists$ and $c_i^2$ for the same $i$.
        \item Type-II Agents contain agents in $a_i^\exists\in I_\exists$ and $c_i^2$ for the same $i$.
        \item Type-III Agents contain the remaining agents in $N_\exists$ and $N_{c^2}$.
    \end{itemize}

    For each feasible tuple, we denote the corresponding allocation $A$ by $$A\equiv\left(H=(H_\exists\cup H_\forall\cup H_b\subseteq N_\exists\cup N_\forall\cup N_b, O_e\subseteq M_e), I_\exists\subseteq N_{\exists}, I_{c^1}\subseteq N_{c^1}, I_{c^3}\subseteq N_{c^3}\right).$$
    The bundle $A_i$ for each agent $i$ in this allocation is constructed as follows. 
    
    \begin{itemize}[align=left, leftmargin=*, labelsep=1em]
        \item[\it{Type-I Agents.}] For each agent $a_i^\exists\in H_\exists$, let $A_{a_i^\exists}=\{g_j^e\in M_e\setminus O_e\mid v_{a_i^\exists}(g_j^e)=1\}\cup\{g_i^a\}$.
        Verbally, each Type-I agent in $N_a$ will get all her positively valued edge items except those three involved in the $3$-regular subgraph, and additionally get $g_i^a$ to meet feasibility.
        Let $A_{c_i^2}=\{g_i^2\}$.

        \item[\it{Type-II Agents.}] For each agent $a_i^\exists\in I_\exists$, let $A_{a_i^\exists}=\{g_j^e\in M_e\mid v_{a_i^\exists}(g_j^e)=1\}$.
        Verbally, she will get all her positively valued edge items.
        Let $A_{c_i^2}=\{g_i^a\}$.
        
        \item[\it{Type-III Agents.}] For each agent $a_i^\exists\in N_\exists\setminus(H_\exists\cup I_\exists)$, let $A_{a_i^\exists}=\{g_j^e\in M_e\mid v_{a_i^\exists}(g_j^e)=1\}$, and $A_{c_i^2}=\{g_i^a, g_i^2\}$.
        
        \item[$N_\forall$.] For each agent $a_i^\forall\in H_\forall$, let $A_{a_i^\forall}=\{g_j^e\in M_e\setminus O_e\mid v_{a_i^\forall}(g_j^e)=1\}\cup\{g_i^\forall\}$, that is, she will get all her positively valued edge items except those three involved in the $3$-regular subgraph, and additionally get $g_i^\forall$ to meet feasibility.
        For each $a_i\in N_\forall\setminus H_\forall$, let $A_{a_i^\forall}=\{g_j^e\in M_e\mid v_{a_i^\forall}(g_j^e)=1\}$, that is, she will get all her positively valued edge items.
        
        \item[$N_b$.] For each agent $b_i\in H_b$, let $A_{b_i}=\{g_j^e\in O_e\mid v_{b_i}(g_j^e)=1\}$, that is, she will get the three edge items involved in the $3$-regular subgraph left by $H_\exists\cup H_\forall$.
        For each agent $b_i\in N_b\setminus H_b$, let $A_{b_i}=\{g_i^b\}$.
        
        \item[$N_{c^1}.$] For each agent $c_i^1\in I_{c^1}$, let $A_{c_i^1}=\{g_i^\exists\}$.
        For each agent $c_i^1\in N_{c^1}\setminus I_{c^1}$, let $A_{c_i^1}=\{g_i^1\}$.
        
        \item[$N_{c^3}.$] For each agent $c_i^3\in I_{c^3}$, let $A_{c_i^3}=\emptyset$.
        For each agent $c_i^3\in N_{c^3}\setminus I_{c^3}$, let $A_{c_i^3}=\{g_i^3\}$.
        
        \item[$z_1$.] Let $A_{z_1}=\{g_i^b\mid b_i\in H_b\}\cup\{g_i^\exists\mid c_i^1\in N_{c^1}\setminus I_{c^1}\}\cup\{g_i^\forall\mid a_i^\forall\in N_\forall\setminus H_\forall\}$.
        
        \item[$z_2$.] Let $A_{z_2}=\{g_i^2\mid a_i^\exists\in I_\exists\}\cup\{g_i^1\mid c_i^1\in I_{c^1}\}\cup\{g_i^3\mid c_i^3\in I_{c^3}\}$.
    \end{itemize}
    \medskip

    The following proposition shows the correspondence between a desired allocation and a feasible tuple.
    \begin{proposition}\label{prop:tuple-alloc}
        An allocation is feasible and non-wasteful if and only if it can be induced by a feasible tuple.
    \end{proposition}
    \begin{proof}
        It is easy to verify that a feasible tuple is sufficient for a feasible and non-wasteful allocation, as each item is allocated to an agent who positively values it, and each agent's utility meets the constraints.
        We show the necessity in the following by constructing a feasible tuple from a feasible and non-wasteful allocation.
    
        In a feasible and non-wasteful allocation, as $u_{a_i}^+=u_{a_i}^-=d_{a_i}$ for each agent $a_i\in N_a$, agent $a_i$ needs to receive either all items from $M_e$ with value $1$, or $d_{a_i}-3$ items from $M_e$ with value $1$ together with item $g_i^a$ (resp., $g_i^{\forall}$) if $a_i\in N_\exists$ (resp., $a_i\in N_\forall$).
        In the latter case, the three items will be allocated to agents in $N_b$.
        Denote the set of agents $a_i^\exists\in N_\exists$ that receive $g_i^a$ as $H_\exists$, the set of agents $a_i^\forall\in N_\forall$ that receive $g_i^{\forall}$ as $H_\forall$, and the set of items that are not allocated to $N_a$ as $O_e$.
        As $u_{b_i}^+=u_{b_i}^-=3$ for each agent $b_i\in N_b$, $b_i$ needs to receive either $g_i^b$, or three items from $M_e$ with value $1$.
        Denote the set of agents $b_i\in N_b$ that do not receive $g_i^b$ as $H_b$.
        It follows straightforwardly that $|H_b|=|H_\exists\cup H_\forall|$ and $H=(H_\exists\cup H_\forall\cup H_b, O_e)$ forms a $3$-regular subgraph of $G$.
    
        As each agent in $N_{c^2}$ is required to receive at least one item, if $g_i^a$ is allocated to $a_i^\exists\in H_\exists$, agent $c_i^2$ will take $g_i^2$.
        For the remaining agents in $N_{c^2}$, each may either receive both items $g_i^a$ and $g_i^2$, or receive $g_i^a$ only and leave $g_i^2$ to agent $z_2$.
        Denote the set of agents $a_i^\exists \in N_\exists$ where $c_i^2$ belongs to the latter case (receive $g_i^a$ only) by $I_\exists$.
    
        Agent $z_1$ has already been allocated $|H_b|$ items from $M_b$ that are not allocated to agents in $N_b$, and $|N_\forall\setminus H_\forall|$ items from $M_\forall$ that are not allocated to agents in $N_\forall$.
        To guarantee that $z_1$ receives $t$ items in total, $z_1$ will take additional $t_1-|H_\exists|$ items from $M_\exists$.
        Consequently, $|H_\exists|$ items from $M_\exists$ will be given to agents in $N_{c^1}$, and $|H_\exists|$ items from $M_{c^1}$ will be given to $z_2$ correspondingly.
        Denote the set of agents in $N_{c^1}$ that receive items from $M_\exists$ by $I_{c^1}$.
        
        Agent $z_2$ has already been allocated $|I_\exists\cup H_\exists|$ items from $M_{c^1}$ and $M_{c^2}$, each with value $3$.
        To guarantee that $z_2$ receives a value of $3t_1$ in total, $z_2$ will take additional $t_1-|I_\exists\cup H_\exists|$ items from $M_{c^3}$, and the remaining items from $M_{c^3}$ will be given to agents in $N_{c^3}$.
        Denote the set of agents in $N_{c^3}$ that receive no item by $I_{c^3}$.
    
        By now, we obtain a feasible tuple $\left(H=(H_\exists\cup H_\forall\cup H_b, O_e), I_\exists, I_{c^1}, I_{c^3}\right)$.
    \end{proof}

    In the next proposition, we show the correspondence between the instance $G$ and agent $z_2$'s bundle in a Pareto-optimal allocation.
    As we have mentioned before, after completing our construction, if agent $z_2$ receives no items from $M_{c^2}$, i.e., $A_{z_2}\cap M_{c^2}=\emptyset$, the instance will fail envy-freeness.
    
    \begin{proposition}\label{prop:feasiblePO-3regular}
        There exists a feasible and Pareto-optimal allocation $(A_1,\ldots,A_n)$ where $A_{z_2}\cap M_{c^2}\neq\emptyset$ 
        if and only if the \neeakrs instance is a \yes instance.
    \end{proposition}
    \begin{proof}
        We first provide a high-level idea of the proof.
        Fix a feasible allocation when given a feasible tuple.
        To see whether the allocation is Pareto-optimal, we consider whether there exists a ``Pareto-improvement cycle'' where, by transferring items along the cycle, some agents' utilities will strictly increase, while others' utilities do not decrease.
        Note that if an agent in $N_{c^2}$ that belongs to Type-III, her allocation cannot be changed without decreasing her utility (since we are considering Pareto-optimality, we restrict ourselves to non-wasteful allocations).
        Therefore, Type-III Agents will not be involved in the Pareto-improvement cycle.
        The Pareto-improvement must involve some Type-II Agents, for otherwise, the social welfare will not increase.
        We will also show that the Pareto-improvement can only involve agents of Type-II, without involving any agents in $N_c^2$ of Type-I.
        Therefore, the existence of such a Pareto-improvement cycle can be roughly viewed as the existence of a $3$-regular subgraph that contains some Type-II agents and some agents belonging to $N_\forall$.
    
        We now provide a formal proof.
        Consider a \yes instance of \neeakrs.
        Without loss of generality, let $U_0=\{u_1^\exists,\ldots,u_{\ell_2}^\exists\}\subseteq U_\exists$ be the non-empty subset with no $3$-regular subgraph in $G$ when combined with any subset of $U_\forall$.
        However, there may exist a subset of $U_0$ that forms a $3$-regular subgraph with some subset of $U_\forall$.
        Let $\{u_1^\exists,\ldots,u_{\ell_1}^\exists\}$ be such a maximal subset in $U_0$, and the $3$-regular subgraph be $$G'=\left(U'\cup V'=\{u_1^\exists,\ldots,u_{\ell_1}^\exists\}\cup\{u_1^\forall,\ldots,u_{\ell_3}^\forall\}\cup\{v_1,\ldots,v_{\ell}\}, E'=\{e_1,\ldots,e_{3\ell}\}\right),$$ where $\ell=\ell_1+\ell_3$ (if there is no subset of $U_0$ that extends to a $3$-regular subgraph, we set $U'=V'=E'=\emptyset$).
        Consider the following allocation $A=(A_1,\ldots,A_n)$ induced by the feasible tuple
        $$A\equiv \left(H, I_\exists=\{a_{\ell_1+1}^\exists,\ldots,a_{\ell_2}^\exists\}, I_{c^1}=\{c_1^1,\ldots,c_{\ell_1}^1\}, I_{c^3}=\{c_{\ell_2+1}^3,\ldots,c_{t_1}^3\}\right)$$ 
        with $H$ corresponds to $G'$
        $$H=\left(\{a_1^\exists,\ldots,a_{\ell_1}^\exists\}\cup\{a_1^\forall,\ldots,a_{\ell_3}^\forall\}\cup\{b_1,\ldots,b_{\ell}\},\{g_1^e,\ldots,g_{3\ell}^e\}\right).$$
        According to Proposition~\ref{prop:tuple-alloc}, allocation $A$ is feasible and non-wasteful.
        Further, as $I_\exists\neq\emptyset$, we have $A_{z_2}\cap M_{c^2}\neq\emptyset$.
    
        We claim that no allocation Pareto-dominates $A$.
        Suppose, for contradiction, that it is Pareto-dominated by a non-wasteful allocation $B=(B_1,\ldots,B_n)$.
        To maintain the utility, each agent $c_i^3\in N_{c^3}\setminus I_{c^3}$ must still receive $g_i^3$ in $B$.
        Each agent $c_i^2$ where $i\in \{1,\ldots,\ell_1\}$ must receive at least item $g_i^2$, and each agent $c_i^2$ where $i\in \{\ell_2+1,\ldots,t_1\}$ must receive both items $g_i^a$ and $g_i^2$.
        Hence, the only potential way to increase social welfare is to transfer some $g_i^a$ from agent $c_i^2$ to $a_i^\exists$ where $a_i^\exists\in I_\exists$, and reallocate other items accordingly.
        Let $x=|\bigcup_{i\in N_{c^2}}A_{c_i^2}\setminus \bigcup_{i\in N_{c^2}}B_{c_i^2}|>0$ (since the social welfare of $B$ should strictly increase compared to $A$) and $y=|\bigcup_{i\in N_{c^2}}B_{c_i^2}\setminus \bigcup_{i\in N_{c^2}}A_{c_i^2}|\ge 0$.
        As each agent $c_i^2$ who loses item $g_i^a$ in $B$ must receive $g_i^2$, it follows that $x\le y$.
        Further, to ensure that the social welfare of the remaining agents $N\setminus N_{c^2}$ does not decrease, it follows that $x\ge y$.
        Therefore, we have $x=y$.
        Moreover, notice that each agent $c_i^2$ in Type-I can only receive item $g_i^a$ from $a_i^\exists$ without giving out item $g_i^2$ to agent $z_2$.
        This implies that there are exactly $x$ agents in $N_{c^2}$ that belong to Type-II, where each agent $c_i^2$ gives item $g_i^a$ to agent $a_i^\exists$ and takes item $g_i^2$ from agent $z_2$.
        Denote this subset by $\{c_{\ell_1+1}^2,\ldots,c_{\ell_1+x}^2\}$, which is a subset of Type-II Agents.
        Hence, the set of agents in $N_{\exists}$ that receive $g_i^a$ in $B$ becomes $\{a_1^\exists,\ldots,a_{\ell_1}^\exists\}\cup\{a_{\ell_1+1}^\exists,\ldots,a_{\ell_1+x}^\exists\}$.
        Further, the social welfare of agents in $N\setminus N_{c^2}$ remains unchanged.
        Consequently, each of the above agents' utility in $B$ remains unchanged, and $B$ is still a feasible allocation.
        According to Proposition~\ref{prop:tuple-alloc}, allocation $B$ is also induced by a feasible tuple containing the $3$-regular subgraph $H'=(\{a_1^\exists,\ldots,a_{\ell_1+x}^\exists\}\cup H'_\forall\cup H'_b, O'_e)$ for some $H'_\forall\subseteq N_\forall, H'_b\subseteq N_b$, and $O'_e\subseteq M_e$.
        However, this contradicts the assumption that $\{u_1^\exists,\ldots,u_{\ell_1}^\exists\}$ is a maximal subset of $U_0$ that contains a $3$-regular subgraph with some subset of $U_\forall$, as $\{u_1^\exists,\ldots,u_{\ell_1}^\exists\}\subset \{u_1^\exists,\ldots,u_{\ell_1+x}^\exists\}$.
    
        We now turn to the other direction.
        Assume the \neeakrs instance is a \no instance.
        Consider any feasible allocation $A=(A_1,\ldots,A_n)$ where $A_{z_2}\cap M_{c^2}\neq\emptyset$, we will show that it is Pareto-dominated by another allocation that does not satisfy the requirement, hence completing the proof.
        By Proposition~\ref{prop:tuple-alloc}, we may assume $A$ is induced by a feasible tuple $\left(H=\left(H_\exists\cup H_\forall\cup H_b, O_e\right), I_\exists, I_{c^1}, I_{c^3}\right)$ where $I_\exists\neq\emptyset$.
        Let $H'_\exists=H_\exists\cup I_\exists$.
        As the instance is a \no instance, there always exists $H'_\forall\subseteq N_\forall$, $H'_b\subseteq N_b$, and $O'_e\subseteq M_e$ such that $H'=(H'_\exists\cup H'_\forall\cup H'_b, O'_e)$ forms a $3$-regular subgraph of $G$.
        Let $I'_{c^1}\subseteq N_{c^1}$ be an arbitrary set with $|I'_{c^1}|=|H'_\exists|$.
        Consider another feasible allocation $B=(B_1,\ldots,B_n)$ induced by the feasible tuple
        $$B\equiv\left(H'=\left(H'_\exists\cup H'_\forall\cup H'_b, O'_e\right), \emptyset, I'_{c^1}, I_{c^3}\right).$$
        Each agent $c_i^2$ where $a_i^\exists\in I_\exists$ is allocated item $g_i^a$ in $A$ and $g_i^2$ in $B$, hence her utility strictly increases by $1$.
        We can verify that the utilities of other agents remain unchanged.
        Therefore, allocation $B$ Pareto-dominates allocation $A$, and $B_{z_2}\cap M_{c^2}=\emptyset$.
    \end{proof}

    Based on the above propositions, we finish our proof by introducing additional agents and items to the above instance $\mathcal{I}$.
    Denote the new instance by $\widetilde{\mathcal{I}}=(\widetilde{N},\widetilde{M},\{\widetilde{v_i}\}_{i\in \widetilde{N}}, \{(\widetilde{u_i}^+, \widetilde{u_i}^-)\}_{i\in \widetilde{N}})$.

    Let $N_f=\{f_1,\ldots,f_{6t_1+1}\}$ be a set of additional agents.
    Let $M_f=\{h^f_1,\ldots, h^f_{(t_1-1)(6t_1+1)}\}$ be a set of additional items such that each agent in $N_f$ has value $1$ to each item in $M_f$.
    Each agent in $N_f$ also has value $1$ to each item in $M_{c^1}\cup M_{c^3}$.
    Let $\widetilde{u_i}^+=\infty$ and $\widetilde{u_i}^-=0$.
    $N_f$ and $M_f$ are constructed to ensure that, if $A_{z_2}\cap M_{c^2}=\emptyset$, then agents in $N_f$ will envy agent $z_2$ and the allocation is no longer envy-free.

    We further introduce additional items to guarantee envy-freeness of agents in $\mathcal{I}$.
    Let $M_{h^2}=\{h_1^2,\ldots,h_{t_1}^2\}$, and $M_{h^3}=\{h_1^3,\ldots,h_{3t_1}^3\}$ respectively be private items to $N_{c^2}$ and $N_{c^3}$.
    In particular, each agent $c_i^2\in N_{c^2}$ has value $1$ to item $h_i^2$; each agent $c_i^3\in N_{c^3}$ has value $1$ to items $h_{3i-2}^3,h_{3i-1}^3$, and $h_{3i}^3$.
    In addition, for each agent $i\in N_a$, if the degree $d_i$ of the corresponding vertex in $G$ is less than $3$, we construct $3-d_i$ items $h^a$ that each has value $1$ only to agent $i$, and set $\widetilde{u_i}^+=\widetilde{u_i}^-=3$.

    For an allocation $\widetilde{A}$ of instance $\widetilde{\mathcal{I}}$ to be envy-free and Pareto-optimal, each agent $c_i^2\in N_{c^2}$ receives item $h_i^2$, each agent $c_i^3\in N_{c^3}$ receives items $\{h_{3i-2}^3,h_{3i-1}^3$, and $h_{3i}^3\}$, and each agent in $N_a$ receives $h^a$ that are only valued by her.
    Each agent in $N_f$ must receive the same value from $M_f\cup M_{c^1}\cup M_{c^3}$, which is at most $t_1-1$.
    In addition, the $(t_1-1)(6t_1+1)$ items in $M_f$ must be allocated among agents in $N_f$.
    Hence, each agent in $N_f$ receives exactly $t_1-1$ items from $M_f$ in any envy-free allocation.
    The remaining agents and items are exactly those in the instance $\mathcal{I}$, and we denote the allocation among them by $A$, which is a partial allocation of $\widetilde{A}$.
    To guarantee that $\widetilde{A}$ is Pareto-optimal, allocation $A$ must also be Pareto-optimal among the agents and items in $\mathcal{I}$.
    To guarantee that $\widetilde{A}$ is feasible and envy-free, allocation $A$ must also be feasible (for agents where $\widetilde{u_i}^+\neq u_i^+$ or $\widetilde{u_i}^-\neq u_i^-$, feasibility of $A$ is defined under $u_i^+$ and $u_i^-$), and satisfy the condition $A_{z_2}\cap M_{c^2}\neq\emptyset$, for otherwise each agent in $N_f$ has value $t_1$ to $A_{z_2}$ hence envies agent $z_2$.
    Conversely, each feasible and Pareto-optimal allocation with $A_{z_2}\cap M_{c^2}\neq\emptyset$ can be extended to a feasible, Pareto-optimal, and envy-free allocation $\widetilde{A}$ by allocating the additional items as described above.
    
    In conclusion, the bounded fair division instance $\widetilde{\mathcal{I}}$ admits a feasible, envy-free, and Pareto-optimal allocation if and only if there exists a feasible and Pareto-optimal allocation $A$ with $A_{z_2}\cap M_{c^2}\neq\emptyset$ in instance $\mathcal{I}$, consequently, if and only if the \neeakrs instance is a \yes instance according to Proposition~\ref{prop:feasiblePO-3regular}.
\end{proof}

In the proof of Lemma~\ref{lem:hard:infty04-bpo}, there are two types of items that have value $3$ to some agents.
The first type of items, denoted by $(3,3)$-item, contains all the items in $M_\exists\cup M_\forall\cup M_b\cup M_{c^1}\cup M_{c^2}\cup M_{c^3}$, each has value $3$ to two agents.
The second type of items, denoted by $(2,3)$-item, contains the items in $M_a$, each has value $3$ to one agent and $2$ to another.

To extend our result to $k=3$-ary valuations, we will introduce two types of gadgets to simulate those items, where each agent's value to each item is at most $2$ in the gadgets.
By applying the gadgets to above \boundefpo-$[\infty,\add,4]$ instance, we get a new \boundefpo-$[\infty,\add,3]$ instance.
Our goal is to show that, the new instance is a \yes instance if and only if the \boundefpo-$[\infty,\add,4]$ instance is a \yes instance, consequently the \neeakrs instance is a \yes instance.

\begin{lemma}\label{lem:hard:infty03-bpo}
    \boundefpo-$[\infty,\add,3]$ is $\sigma2$-hard.
\end{lemma}

\begin{proof}
    We begin by describing the two types of gadgets and their properties.
    Note that in the following part, we may reuse the previously introduced notations, which are independent of the analysis above.
    If we do not explicitly define an agent's value to an item, the value is $0$.

    \begin{figure}[h]
    \centering
    \begin{tikzpicture}[
        node distance=1.5cm and 3cm,
        agent/.style={rectangle, draw=black, rounded corners, minimum size=0.6cm, font=\small, align=center, font=\bfseries},
        item_transfer/.style={->, >={Stealth[length=2mm]},  thick, color=red!80!black,rounded corners=4pt},
        edge/.style={-, color=black,rounded corners=4pt},
        label_text/.style={font=\itshape, align=center, color=black},
        label_text2/.style={font=\itshape, align=center, color=black, font=\small},
        section_box/.style={draw, dashed, rounded corners, inner sep=8pt},
        section_box2/.style={draw, dashed, rounded corners, inner sep=8pt}
    ]

    \node[font=\bfseries] at (-1,8) {$(3,3)$-Gadget};
    \node[agent] (111) at (1, 8) {$1$};
    \node[agent] (112) at (2.9, 8) {$2$};
    \draw[edge] (111) -- (112) node[midway, below] {$g$} node[pos=0.1,above,font=\tiny,color=red!80!black] {3} node[pos=0.9,above,font=\tiny,color=red!80!black] {3};

    \draw[->, >={Stealth[length=2mm]}, very thick] (3.7,8) to (4.2,8);

    \node[agent] (12a) at (7,9) {$a$};
    \node[agent] (12b) at (7,7) {$b$};
    \node[agent] (121) at (5,8) {$1$};
    \node[agent] (122) at (9,8) {$2$};
    
    \draw[edge] (12a) -- (12b) node [pos=0.5, right] {$g_0$} node[pos=0.15,left,font=\tiny,color=red!80!black] {1} node[pos=0.85,left,font=\tiny,color=red!80!black] {1};
    \draw[edge] (121) -- (12a) node [pos=0.6, below] {$g_1$} node[pos=0.1,above,font=\tiny,color=red!80!black] {2} node[pos=0.9,above,font=\tiny,color=red!80!black] {$2$};
    \draw[edge] (121) -- (12b) node [pos=0.6, above] {$g'_1$} node[pos=0.1,below,font=\tiny,color=red!80!black] {1} node[pos=0.9,below,font=\tiny,color=red!80!black] {1};
    \draw[edge] (122) -- (12a) node [pos=0.6, below] {$g_2$} node[pos=0.1,above,font=\tiny,color=red!80!black] {1} node[pos=0.9,above,font=\tiny,color=red!80!black] {1};
    \draw[edge] (122) -- (12b) node [pos=0.6, above] {$g'_2$} node[pos=0.1,below,font=\tiny,color=red!80!black] {2} node[pos=0.9,below,font=\tiny,color=red!80!black] {2};

    \node[font=\bfseries] at (-1,5.5) {$(2,3)$-Gadget};
    \node[agent] (211) at (1, 5.5) {$1$};
    \node[agent] (212) at (2.9, 5.5) {$2$};
    \draw[edge] (211) -- (212) node[midway, below] {$g$} node[pos=0.1,above,font=\tiny,color=red!80!black] {2} node[pos=0.9,above,font=\tiny,color=red!80!black] {3};

    \draw[->, >={Stealth[length=2mm]}, very thick] (3.7,5.5) to (4.2,5.5);

    \node[agent] (22a) at (7,5.5) {$a$};
    \node[agent] (221) at (5.2,5.5) {$1$};
    \node[agent] (222) at (9,5.5) {$2$};
    
    \draw[edge] (221) -- (22a) node [pos=0.5, below] {$g_1$} node[pos=0.1,above,font=\tiny,color=red!80!black] {2} node[pos=0.9,above,font=\tiny,color=red!80!black] {2};
    \draw[edge] (22a) to[bend right=45] node[pos=0.5, above] {$g'_2$} node[pos=0.1,below,font=\tiny,color=red!80!black] {1} node[pos=0.9,below,font=\tiny,color=red!80!black] {1} (222) ;
    \draw[edge] (22a) to[bend left=45] node[pos=0.5, below] {$g_2$} node[pos=0.1,above,font=\tiny,color=red!80!black] {1} node[pos=0.9,above,font=\tiny,color=red!80!black] {2} (222) ;
    
    \end{tikzpicture}
    \caption{Visualization of the gadgets. Vertices represent agents and edges represent items. Items are only positively valued by endpoint agents. }
    \end{figure}
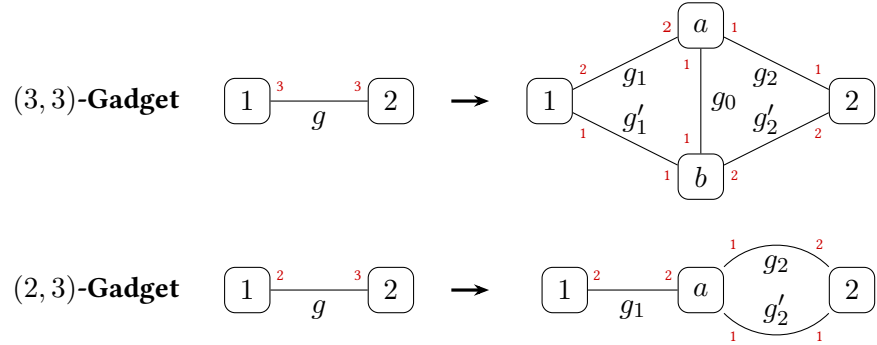

    \paragraph{$(3,3)$-Item Gadget.}
    Assume an item $g$ has value $3$ to both agents $1$ and $2$. We construct two
    new agents $a$ and $b$ with $u_a^+=u_b^+=u_a^-=u_b^-=2$, and five items
    $G=\{g_1,g'_1,g_2,g'_2,g_0\}$. We call agents $a$ and $b$ $(3,3)$-gadget
    agents and the five items $(3,3)$-gadget items. Let
    $v_1(g_1)=2$, $v_1(g'_1)=1$, $v_2(g_2)=1$, and $v_2(g'_2)=2$.
    Let $v_a(g_1)=2$, $v_a(g_2)=1$, $v_b(g'_1)=1$, and $v_b(g'_2)=2$.
    Let $v_a(g_0)=v_b(g_0)=1$. We simulate the case where agent $1$ receives
    item $g$ by letting agent $1$ receive $g_1$ and $g'_1$, agent $a$ receives
    $g_2$, and $g_0$, and $b$ receives $g'_2$. Similarly, we simulate the case
    where agent $2$ receives item $g$ by letting agent $2$ receive $g_2$ and
    $g'_2$, $a$ receives $g_1$, and $b$ receives $g'_1$ and $g_0$.
    We refer to these two allocations as \emph{simulating allocations}.

    It is easy to see that the above two allocations are the only feasible and non-wasteful allocations for the $(3,3)$-gadget agents in this gadget.
    For any other non-wasteful allocation where the gadget agents meet their utility lower bounds, there is always some gadget agent whose utility is strictly more than $2$, which violates the requirement on her utility upper bound.

    In addition, in a non-wasteful allocation, it is guaranteed that if one of agents $1$ and $2$ receives an item with value $2$ from the gadget, then the other agent receives no item from this gadget when requiring each gadget agent to meet the utility lower bound.
    Suppose agent $1$ receives item $g_1$, then agent $a$ needs to receive items $g_2$ and $g_0$ to guarantee a utility of $2$.
    Consequently, agent $b$ needs to receive item $g'_2$, for otherwise, agent $b$'s utility is at most $1$.
    Item $g'_1$ can be allocated to either agent $1$ or $b$.
    Thus, nothing is left for agent $2$ from the gadget.

    \paragraph{$(2,3)$-Item Gadget.} Assume an item $g$ has value $2$ to agent $1$ and value $3$ to agent $2$.
    We construct one new agent $a$ and three items $G=\{g_1,g_2,g'_2\}$.
    We call agent $a$ \emph{$(2,3)$-gadget agent} and the three items \emph{$(2,3)$-gadget items}.
    Let $v_1(g_1)=2$, $v_2(g_2)=2$ and $v_2(g'_2)=1$.
    Let $v_{a}(g_1)=2$, and $v_{a}(g_2)=v_{a}(g'_2)=1$, and $u_{a}^+=u_{a}^-=2$.
    In this case, agent $a$ needs to either receive $g_1$, or both $g_2$ and $g'_2$ to guarantee feasibility.
    We simulate the case where agent $1$ receives item $g$ by letting agent $1$ receive $g_1$ while agent $a$ receives $g_2$ and $g'_2$.
    Similarly, we simulate the case where agent $2$ receives item $g$ by letting agent $a$ receive $g_1$ while agent $2$ receives $g_2$ and $g'_2$.
    \medskip

    Based on the above gadgets, we first replace each item that has value $3$ in the \boundefpo-$[\infty,\add,4]$ instance $\mathcal I$ constructed before Proposition~\ref{prop:tuple-alloc} by the gadgets to obtain a \boundefpo-$[\infty,\add,3]$ instance, denoted by $\mathcal{I}^\dag$.
    We will show the feasibility and Pareto-optimality under instance $\mathcal{I}^\dag$.
    In the following, we use ``original agents'' to refer to the agents in $\mathcal{I}$.
    Write $N^\dagger$ and $M^\dagger$ for the agent and item sets of $\mathcal I^\dagger$, respectively, and retain the original agents' utility bounds.

    Assume that $\mathcal{I}$ admits a feasible and Pareto-optimal allocation $A$, induced by a feasible tuple 
    $$\left(H=\left(H_\exists\cup H_\forall\cup H_b, O_e\right), I_\exists, I_{c^1}, I_{c^3}\right)$$ where $I_\exists\neq\emptyset$, which indicates $A_{z_2}\cap M_{c^2}\neq\emptyset$.
    We construct an allocation $A^\dag$ for $\mathcal{I}^\dag$ by letting the gadget items be allocated according to the above description of our simulating allocation.
    As $A$ is feasible, it is straightforward that $A^\dag$ is also feasible and each gadget agent is envy-free.
    To verify it is also Pareto-optimal, assume for contradiction that it is Pareto-dominated by another allocation $B^\dag$.
    We assume without loss of generality that $B^\dag$ is non-wasteful (otherwise, any item allocated to an agent who values it at $0$ can be reassigned to other agents).
    Notice the only possible way to increase social welfare is to reallocate some $(2,3)$-gadget items from $(2,3)$-gadget agents that have value $1$, to the original agents that have value $2$ to them.
    Those original agents correspond to $I_\exists$ and belong to Type-II Agents.
    Let $N'$ denote the set of agents in $\mathcal I^\dagger$ except for agents
    \[
    \begin{aligned}
    N_{c^2}\ &\cup\
    \{\text{$(2,3)$-gadget agents constructed from $M_a$}\}\\
    &\cup\
    \{\text{$(3,3)$-gadget agents constructed from $M_{c^2}$}\}.
    \end{aligned}
    \]
    We focus on the social welfare of agents in $N'$. 
    Let $x$ be the utility gain to $N'$ from the items that are reallocated from agents in $N^\dagger\setminus N'$ to $N'$, that is,
    \[
    x=\sum_{i\in N'}\ \sum_{g\in B_i^\dagger\setminus
                        (\bigcup_{j\in N'}A_j^\dagger)}v_i(g).
    \]
    Let $y$ be the utility loss to $N'$ by the items that are reallocated from $N'$ to agents in $N^\dagger\setminus N'$, that is,
    \[
    y=\sum_{i\in N'}\ \sum_{g\in A_i^\dagger\setminus
                        (\bigcup_{j\in N'}B_j^\dagger)}v_i(g).
    \]
    Every item whose positively valuing agents all belong to $N'$ has the same value to those agents. Thus, the change in the total utility of $N'$ is exactly $x-y$.
    To guarantee the social welfare of $N'$ does not decrease, we have $x\ge y$. 
    On the other hand, we need to guarantee that when reallocating items based on allocation $A^\dagger$, the utility of each agent in $N_{c^2}$ does not decrease. 
    We separately consider the agents in $N_{c^2}$ that belong to Type-I and Type-II whose items get reallocated.

\begin{itemize}
        \item When a Type-II Agent $c_i^2$ gives away a $(2,3)$-gadget item with value $2$ constructed from item $g_i^a$, the increase in $x$ is at most $3$ when the $(2,3)$-gadget agent gives away both items to agent $a_i^\exists$. 
        She needs to receive a $(3,3)$-gadget item with value $2$ from the $(3,3)$-gadget constructed from item $g_i^2$.
        However, as we have shown, when $c_i^2$ receives a $(3,3)$-gadget item with value $2$, then $z_2$ receives no item from this gadget when we further require that the utilities of the two $(3,3)$-gadget agents do not decrease.
        Therefore, the loss of utility for agent $z_2$ is at least $3$.
        If instead $c_i^2$ keeps the $(2,3)$-gadget item worth $2$ to her, the $(2,3)$-gadget agent must retain the other two items, so $a_i^\exists$ receives nothing from this gadget; any loss of $z_2$ can only decrease $x-y$.

        \item A Type-I Agent $c_i^2$ can at most receive a value of $2$ from the $(2,3)$-gadget item constructed from item $g_i^a$.
        This requires the $(2,3)$-gadget agent to receive both items from $a_i^\exists$ to guarantee her utility remains $2$, which results in a utility loss of $3$ for agent $a_i^\exists$.
        Agent $c_i^2$, on receiving the $(2,3)$-gadget item with value $2$, may give away at most one item in the $(3,3)$-gadget item constructed from item $g_i^2$.
        When requiring the utilities of the two $(3,3)$-gadget agents do not decrease, agent $z_2$ receives at most an additional value of $1$ from this gadget.
        If $c_i^2$ does not receive the $(2,3)$-gadget item, she must retain both $(3,3)$-gadget items worth $3$ to her, and $z_2$ receives nothing from that gadget.
        
        \item As each agent of Type-III reaches her maximum possible utility, her bundle cannot be changed, and the three associated gadget agents must also retain their simulating bundles.
    \end{itemize}

    Combining the three cases, we have $x\le y$. Combining with $x\ge y$, we have $x=y$. 
    Further, we identify that the reallocation involves the agents in $N_{c^2}$ of Type-II, and involves no agents in $N_{c^2}$ of Type-I or Type-III.

    Since $x=y$, every agent in $N'$ receives exactly the same utility in $B^\dagger$ as in $A^\dagger$. 
    In particular, every gadget agent in $N'$ receives exactly $2$, so the corresponding gadgets have simulating allocations. 
    This does not yet imply that the gadget agents outside $N'$ receive exactly $2$. 
    We therefore normalize the remaining gadgets before projecting back to $\mathcal I$.

    Each pair of gadgets constructed from $g_i^a$ and $g_i^2$ contributes a non-positive amount to $x-y$, by the three cases above. 
    Hence, equality must hold separately for every such pair. 
    For Type-I and Type-III Agents, equality forces the pair to remain unchanged. 
    For a Type-II Agent $c_i^2$, there are precisely two possibilities for the items received by $a_i^\exists$ and $z_2$: either the pair remains unchanged, or $a_i^\exists$ receives both of her $(2,3)$-gadget items, of total value $3$, and $z_2$ receives neither of her $(3,3)$-gadget items. 
    Call the latter pair \emph{switched}. 
    There must be at least one switched pair; otherwise all agents outside $N'$ would also retain their original utilities, contradicting strict Pareto domination.

    Consider a switched pair. 
    Use the local notation of the $(3,3)$-gadget, with $c_i^2$ playing the role of agent $1$ and $z_2$ the role of agent $2$. 
    Agent $c_i^2$ must receive $g_1$. 
    The utility lower bounds of the gadget agents then force $a$ to receive $g_2$ and $g_0$, and $b$ to receive $g'_2$. 
    The remaining item $g'_1$ is allocated either to $c_i^2$ or to $b$. 
    In the former case, the gadget already has its simulating allocation. 
    In the latter case, $c_i^2$ receives utility $2$ and $b$ receives utility $3$; move $g'_1$ from $b$ to $c_i^2$. 
    Their utilities become $3$ and $2$, respectively, while every other utility is unchanged.
    The $(2,3)$-gadget in this pair is already in its simulating allocation, with its gadget agent receiving the item worth $2$ to her.

    Apply this normalization to every switched pair and denote the resulting allocation by $\overline B^\dagger$. 
    It need not Pareto-dominate $B^\dagger$, but it still Pareto-dominates $A^\dagger$: all agents in $N'$ retain their utilities from $A^\dagger$, all gadget agents receive exactly $2$, and each switched Type-II Agent $c_i^2$ receives $3$ instead of her original utility $2$. 
    All gadgets now have simulating allocations.
    Consequently, $\overline B^\dagger$ simulates an allocation $B$ in $\mathcal I$ that Pareto-dominates $A$, contradicting the assumption that $A$ is Pareto-optimal.

    The other direction is more straightforward.
    Assume that $\mathcal{I}^\dag$ admits a feasible and Pareto-optimal allocation $A^\dag$ such that agent $z_2$ receives at least one item from the $(3,3)$-gadgets constructed from items in $M_{c^2}$.
    As a feasible and non-wasteful allocation within each gadget simulates an allocation of the original item, we can construct an allocation $A$ for $\mathcal{I}$ by letting each original agent receive the original items according to the simulation.
    Each original agent receives the same utility in $A$ as in $A^\dag$, thus $A$ is feasible.
    If $A$ is Pareto-dominated by another non-wasteful allocation $B$, then we can simulate the same Pareto-improvement in $\mathcal{I}^\dag$ by a simulating allocation $B^\dag$.
    This contradicts the assumption that $A^\dag$ is Pareto-optimal.

    We have obtained that $\mathcal{I}^\dag$ admits a feasible and Pareto-optimal allocation $A^\dag$ such that agent $z_2$ receives at least one item from the $(3,3)$-gadgets constructed from items in $M_{c^2}$, if and only if $\mathcal{I}$ admits a feasible and Pareto-optimal allocation $A$ such that agent $z_2$ receives at least one item from $M_{c^2}$.
    Similar to the last step in the proof of Lemma~\ref{lem:hard:infty04-bpo}, we introduce a set $N_f$ of $6t_1+1$ additional agents that has value $1$ to an additional set of $(3t_1-1)(6t_1+1)$ items.
    Each agent in $N_f$ also positively values items constructed from $M_{c^1}$ and $M_{c^3}$ the same as $z_2$.
    For these agents, set the lower and upper utility bounds to $0$ and $\infty$, respectively. 
    All other values of these agents, and all existing agents' values for the additional items, are $0$.

    The common value of the agents in $N_f$ for the existing items is $6t_1<|N_f|$. 
    Therefore, envy-freeness forces them to receive the same integer utility, at most $3t_1-1$.
    Pareto-optimality forces all $(3t_1-1)(6t_1+1)$ additional items to be allocated among them, so each receives exactly $3t_1-1$ of these items and no existing item that she values positively.

    Therefore, a feasible, envy-free, and Pareto-optimal allocation of the augmented instance exists only if $z_2$ receives at least one item from the $(3,3)$-gadgets constructed from items in $M_{c^2}$.

    We can further construct private items that are positively valued only by their owner with value $1$, and will always be allocated to their owner, to guarantee that each agent in $N^\dagger$ is envy-free as long as the $A^\dagger$ is feasible to her.
    Explicitly, for each $i\in N^\dagger$, let $q_i=v_i(M^\dagger)$ and add $q_i$ private items for $i$. 
    Replace her lower bound by $u_i^-+q_i$ and her upper bound by $u_i^++q_i$, with $\infty+q_i=\infty$. 
    These private items are worth $0$ to every other agent, including the agents in $N_f$.
    In a Pareto-optimal allocation they all go to their owner. 
    Thus, the shifted bounds enforce exactly the original bounds on the restriction to $\mathcal I^\dagger$, and the private utility $q_i$ prevents $i$ from envying any other bundle. 
    The number of added items is polynomial.

    For completeness, let $A^\dagger$ be a feasible and Pareto-optimal allocation with the required item from an $M_{c^2}$-gadget in $z_2$'s bundle.
    Extend it by assigning all private items to their owners and giving each agent in $N_f$ exactly $3t_1-1$ of the additional items.
    The allocation is feasible and envy-free. Indeed, the gadgets simulate original items, so an agent in $N_f$ values $z_2$'s bundle at most $3t_1-3$, and any other agent's bundle at most $2\le 3t_1-1$.
    To verify Pareto-optimality, consider a non-wasteful Pareto-improvement.
    The private items remain with their owners. 
    If any item of $M^\dagger$ is allocated to an agent in $N_f$, returning all such items to positively valuing agents in $N^\dagger$ would produce a Pareto-improvement of $A^\dagger$. 
    If no such item is allocated to $N_f$, no original utility can strictly increase, and the fixed total value of the additional items prevents a strict improvement for $N_f$. 
    Both cases contradict a Pareto-improvement.
    Conversely, restricting any feasible, envy-free, and Pareto-optimal allocation of the augmented instance to $\mathcal I^\dagger$ yields a feasible and Pareto-optimal allocation with the required item in $z_2$'s bundle: the agents in $N_f$ take no existing items, and any improvement of the restriction could be extended by leaving all additional items unchanged.

    According to Proposition~\ref{prop:feasiblePO-3regular}, we conclude that, the \boundefpo-$[\infty,\add,3]$ instance admits a feasible, envy-free, and Pareto-optimal allocation if and only if the \neeakrs instance is a \yes instance.
\end{proof}
    
Finally, we show the $\sigma2$-hardness of \efpo-$[\infty,\add,3]$ by further introducing two types of gadgets to simulate the feasibility requirement of $u_i^+$ and $u_i^-$.
In particular, the instance after applying the gadgets admits an envy-free allocation only when the instance \boundefpo-$[\infty,\add,3]$ admits a feasible and Pareto-optimal allocation.
\begin{lemma}\label{lem:hard:infty03-po}
    \efpo-$[\infty,\add,3]$ is $\sigma2$-hard.
\end{lemma}
\begin{proof}
    We begin by introducing two types of gadgets, referred to as the upper-bound gadget and the lower-bound gadget.
    The purpose of these gadgets is to encode the explicit feasibility constraints in the bounded fair division instance constructed in the proof of Lemma~\ref{lem:hard:infty03-bpo}, so that any Pareto-optimal and envy-free allocation will necessarily satisfy those constraints by construction.
    Let the bounded fair division instance in the proof of Lemma~\ref{lem:hard:infty03-bpo} be $\mathcal{I}=(N,M,\{v_i\}_{i\in N}, \{(u_i^+, u_i^-)\}_{i\in N})$.

    \paragraph{Upper Bound Gadget.} Consider agent $i$ with $u_i^+ = x$ for some positive integer $x$, where $x$ is polynomial in $n$ and $m$.
    We only need to consider the case that $v_i(M) = t\ge x$, for otherwise, we can directly remove the constraint.
    Note that as $v_i$ is additive and $v_i(g)\in\{0,1,2\}$ for each item $g\in M$ in Lemma~\ref{lem:hard:infty03-bpo}, it follows that $t$ is polynomial in $n$ and $m$.
    We introduce $t+1$ additional agents $\{a_1^+,\ldots,a_{t+1}^+\}$ and $(t+1)x$ additional items $\{g_1^+,\ldots,g_{(t+1)x}^+\}$.
    Each additional item is only positively valued by the additional agents with value $1$.
    For each item $g\in M$, let $v_{a_j^+}(g)=v_i(g)$ for each $j\in\{1,\ldots,t+1\}$.
    The value between other agents and items remains unchanged as in the original instance.
    
    In any envy-free allocation, the additional agents need to receive the same value.
    Combined with the fact that each additional agent's total value to the new instance is $(t+1)x+t$, each additional agent will receive a value at most $x$.
    In any Pareto-optimal allocation, the additional items will be allocated to the additional agents.
    Therefore, each additional agent receives exactly $x$ items from the additional items in any envy-free and Pareto-optimal allocation.
    To ensure the envy-freeness between an additional agent and agent $i$, the items allocated to agent $i$ can have a total value of at most $x$.
    Hence, the upper-bound constraint for agent $i$ is enforced.

    \paragraph{Lower Bound Gadget.} Consider agent $i$ with $u_i^- = y$ for some positive integer $y$, where $y$ is polynomial in $n$ and $m$.
    We introduce $y$ additional agents $\{a_1^-,\ldots,a_{y}^-\}$ and $y^2$ additional items $\{g_1^-,\ldots,g_{y^2}^-\}$.
    Each additional item is positively valued by both agent $i$ and each of the additional agents with value $1$.
    The value between other agents and items remains unchanged as in the original instance.

    In any envy-free allocation, the additional agents need to receive the same number of items among the additional items.
    Moreover, if each additional agent receives fewer than $y$ items, agent $i$ will receive the remaining additional items due to Pareto-optimality, which is at least $y$, violating the envy-freeness between agent $i$ and the additional agents.
    Therefore, each additional agent receives exactly $y$ items from the additional items in any envy-free and Pareto-optimal allocation.
    To ensure that agent $i$ does not envy the additional agents, agent $i$ must receive a value of at least $y$ from $M$, in accordance with the lower bound constraint.

    \medskip
    We now resume our proof.
    The constructed fair division instance to prove Lemma~\ref{lem:hard:infty03-po} is based on $\mathcal{I}$ in Lemma~\ref{lem:hard:infty03-bpo}, where we attach an upper-bound gadget and a lower-bound gadget to each agent $i\in N$ and remove her utility constraints.
    According to our analysis above, a Pareto-optimal and envy-free allocation in the new instance must satisfy the feasibility constraints in $\mathcal{I}$.
    Conversely, any feasible, Pareto-optimal, and envy-free allocation in $\mathcal{I}$ can be extended to a Pareto-optimal and envy-free allocation in the new instance by allocating the additional items to the additional agents evenly.
    Hence, the new instance admits a Pareto-optimal and envy-free allocation if and only if $\mathcal{I}$ in Lemma~\ref{lem:hard:infty03-bpo} admits a feasible, Pareto-optimal, and envy-free allocation.
    This concludes the $\sigma2$-hardness of \efpo-$[\infty,\add,3]$.
\end{proof}

\section{Proofs for Other Results in Table~\ref{tab:results}}
\label{sect:formal}
In this section, we formally prove all remaining results in Table~\ref{tab:results}.
In Sect.~\ref{sect:member}, we prove the memberships for the problems in the respective complexity classes.
Sect.~\ref{sect:hardness} deals with hardness.
The results are combined in Sect.~\ref{sect:combine}.

\subsection{Memberships in Complexity Classes}
\label{sect:member}
First of all, all the problems studied in this paper are in $\sigma2$: the existential certificate is an envy-free allocation, and the for-all certificate is another allocation with potentially larger social welfare or potentially Pareto-dominates the envy-free allocation.

The following four results are easy to see.

\begin{lemma}\label{lem:member:infty0infty-msw}
    \efmsw-$[\infty,\add,\infty]$ is in $\classNP$.
\end{lemma}
\begin{proof}
    For instances with additive valuations, an allocation maximizes the social welfare if and only if each item $g$ is allocated to an agent $i$ such that $v_{ig}$ is maximum among $\{v_{1g},\ldots,v_{ng}\}$.
    This property can obviously be checked in polynomial time.
    Therefore, a certificate encoding an allocation is an $\classNP$ certificate, as both envy-freeness and social welfare optimality can be checked in polynomial time. 
\end{proof}

\begin{lemma}\label{lem:member:20infty-po}
    \efpo-$[2,\add,\infty]$ is in $\classNP$.
\end{lemma}
\begin{proof}
    For two agents, an allocation $(A_1,A_2)$ is envy-free if and only if $v_1(A_1)\geq\frac12 v_1(M)$ and $v_2(A_2)\geq\frac12v_2(M)$.
    Any Pareto-improvements still maintain $v_1(A_1)\geq\frac12 v_1(M)$ and $v_2(A_2)\geq\frac12v_2(M)$, and thus still envy-free.
    Therefore, if there is an envy-free allocation, there is an envy-free allocation that is Pareto-optimal.
    The $\classNP$ certificate can just be the allocation, as we only need to check envy-freeness, which can be done in polynomial time.
\end{proof}

\begin{lemma}\label{lem:member:infty1k-msw}
    For any constant $k\geq2$, \efmsw-$[\infty,\mono,k]$ is in $\theta2$.
\end{lemma}
\begin{proof}
    We will show that we can decide \efmsw-$[\infty,\mono,k]$ by non-adaptively applying an $\classNP$ oracle for polynomially many times.
    For $n$ agents and $m$ items, we have $v_i(M)\leq mk$, so the maximum possible social welfare is $nmk$.
    For each $t=0,1,\ldots,nmk$, we make two queries to the oracle: 1) does there exist an allocation with social welfare $t$, and 2) does there exist an envy-free allocation with social welfare $t$.
    These are clearly $\classNP$ queries, as the corresponding allocation certifies the \yes instance.
    After making $2mnk+2$ queries (which is polynomial in $n$ and $m$), we can easily check in $O(mn)$ time if there is a non-envy-free allocation that has a strictly higher social welfare than that of all envy-free allocations.
\end{proof}

\begin{lemma}\label{lem:member:infty1infty-msw}
    \efmsw-$[\infty,\mono,\infty]$ is in $\delta2$.
\end{lemma}
\begin{proof}
    We will show that we can decide \efmsw-$[\infty,\mono,\infty]$ by adaptively applying an $\classNP$ oracle for polynomially many times.
    Let $T$ be $\sum_{i\in N}v_i(M)$.
    By our assumptions of valuation functions, $\log_2T$ is bounded by a polynomial in $m$ and $n$.
    The idea is similar to the proof of Lemma~\ref{lem:member:infty1k-msw}, but with a binary search technique.
    For a value $t\in\{0,1,\ldots,T\}$, we make two queries: 1) does there exist an allocation with social welfare \emph{at least} $t$, and 2) does there exist an envy-free allocation with social welfare \emph{at least} $t$.
    If the answers are both ``yes'', we only need to make queries within $\{t+1,\ldots,T\}$; if the answers are both ``no'', we only need to make queries within $\{0,1,\ldots,t-1\}$; if the answers to the two questions are different, it must be ``yes'' for the first query and ``no'' for the second query, then we know this is a \no instance.
    If the binary search completes with the same answer for both queries in every round, then this is a \yes instance.
\end{proof}

The next two lemmas are based on the value truth table (Definition~\ref{def:truthtable}).

Recall that the size of the table is polynomial in $m$ and $n$ if $n$ is a constant and valuations are $k$-ary.
It does not matter if the valuation functions are additive or not, as $T\leq mk$ holds for $k$-ary valuations.
If the truth table is constructed, checking the existence of an envy-free Pareto-optimal can obviously be done in polynomial time:
the index of the cell already tells if the corresponding allocation is envy-free; for each ``envy-free index'' that has value ``true'', it suffices to check the Boolean value of each cell that represents a Pareto-dominating allocation to see if any of them has value ``true''.
As a remark, for each cell that has the value ``true'', it may correspond to more than one allocation.

\begin{lemma}\label{lem:member:n0k-po}
    For any constant $n\geq2$ and $k\geq2$, \efpo-$[n,\add,k]$ and \efmsw-$[n,\add,k]$ are in $\classP$.
\end{lemma}
\begin{proof}
    We can construct the valuation truth table $\mathcal{T}$ in polynomial time by the standard dynamic programming method which has also been applied in other fair division papers~\citep{aziz2023computing,han2023average,bu2025approximability}.
    Then it is routine to check the existence of envy-free and Pareto-optimal allocations.
\end{proof}

\begin{lemma}\label{lem:member:n1k-po}
    For any constant $n\geq2$ and $k\geq2$, \efpo-$[n,\mono,k]$ are in $\theta2$.
\end{lemma}
\begin{proof}
    Since the valuation truth table $\mathcal{T}$ has a polynomial size, we can use an $\classNP$ oracle to get the value of each cell in $\mathcal{T}$, and we have used the oracle for polynomially many times and non-adaptively.
    Each query is an $\classNP$ query as an allocation that agrees with the index of the cell gives a certificate.
    Then it is routine to check the existence of envy-free and Pareto-optimal allocations.
\end{proof}

\subsection{Hardness Results}
\label{sect:hardness}

\begin{lemma}\label{lem:hardness:60infty}
    \efpo-$[6,\add,\infty]$ is $\sigma2$-hard.
\end{lemma}
\begin{proof}
    We reduce the problem from \easubsetsum (Problem~\ref{problem:easubsetsum}).
    Given a \easubsetsum instance written as $(V_\forall=\{x_1,\ldots,x_m\},V_\exists=\{x_{m+1},\ldots,x_n\},T)$ (where $1<m<n$), we construct the following fair division instance.
    We assume without loss of generality that $0<T<\sum_{i=1}^nx_i$.
    
    Let $W=10\cdot \sum_{i=1}^nx_i$ be a sufficiently large integer.
    For each $i=1,\ldots,n$, let $y_i=3^i\cdot W$.
    Let $Y_\exists=\sum_{i=m+1}^ny_i$.
    Let $Y_\forall=\sum_{i=1}^my_i$.
    There are six agents: $a_\exists,a_\exists',a_\forall,a_\forall',a_{po},a_{po}'$.
    There are $2n+6$ items: $\{c_1,\neg c_1,c_2,\neg c_2,\ldots,c_n,\neg c_n\}\cup\{g_\exists,g_\exists^-,g_\forall,g_T,g_{po},g_{poT}\}$.
    Let $S_i$ be the item-set $\{c_i,\neg c_i\}$ for each $i=1,\ldots,n$.
    The valuation functions of the six agents are defined below:
    \begin{itemize}
        \item[$a_\exists$:] The value for $g_\exists$ is $Y_\exists+W$, and the value for $g_\exists^-$ is $W$. For each $i=m+1,\ldots,n$, the value for each $c_i$ is $y_i+x_i$ and the value for each $\neg c_i$ is $y_i$. The value for each remaining item is $0$.
        \item[$a_\exists'$:] The value for $g_\exists$ is $Y_\exists+W$, and the value for each remaining item is $0$.
        \item[$a_\forall$:] The value for $g_\forall$ is $Y_\exists+2Y_\forall$. For each $i=1,\ldots,n$, the value for each $c_i$ is $y_i+x_i$, and the value for each $\neg c_i$ is $y_i$. The value for $g_T$ is $Y_\exists+Y_\forall+T$. The value for each remaining item is $0$.
        \item[$a_\forall'$:] The value for $g_\forall$ is $Y_\exists+2Y_\forall$, and the value for each remaining item is $0$.
        \item[$a_{po}$:] For each $i=1,\ldots,n$, the value for each $c_i$ is $y_i+x_i$, and the value for each $\neg c_i$ is $y_i$. The value for $g_T$ is $Y_\exists+Y_\forall+T-\frac13$ (we use the fractional number $\frac13$ for the ease of analysis; we can rescale the whole instance to make the valuations integral). The value for $g_{po}$ is $10(Y_\exists+2Y_\forall+W)$. The value for $g_{poT}$ is exactly the value of $g_{po}$ plus the value of $g_T$. The value for each remaining item is $0$.
        \item[$a_{po}'$:] Exactly the same valuation function as $a_{po}$.
    \end{itemize} 
    Notice that this is almost a \emph{restrictive additive} valuation profile: each item's value is either $0$ or some pre-defined value, with the only exception of $g_T$.

    [Completeness] If the \easubsetsum instance is a \yes instance, there exists $U_\exists\subseteq V_\exists$ such that for all $U_\forall\subseteq V_\forall$ we have $\sum_{x_i\in U_\exists}x_i+\sum_{x_j\in U_\forall}x_j\neq T$.
    We show that the following allocation is envy-free and Pareto-optimal.
    For each $i=m+1,\ldots,n$, if $x_i\in U_\exists$, include $\neg c_i$ to agent $a_\exists$'s bundle, and include $c_i$ to agent $a_\forall$'s bundle; if $x_i\notin U_\exists$, then include $c_i$ to agent $a_\exists$'s bundle and include $\neg c_i$ to agent $a_\forall$'s bundle.
    For each $i=1,\ldots,m$, include both $c_i$ and $\neg c_i$ to agent $a_\forall$'s bundle.
    Further include $g_\exists^-$ to agent $a_\exists$'s bundle.
    For the remaining five items $g_\exists,g_\forall,g_T,g_{po},g_{poT}$, allocate $g_\exists$ to $a_\exists'$, allocate $g_\forall$ to $a_\forall'$, allocate $g_T$ and $g_{po}$ to $a_{po}$, and allocate $g_{poT}$ to $a_{po}'$.
    A careful check can verify that the allocation is envy-free.
    We then show that the allocation is Pareto-optimal.
    
    Firstly, the values for both agents $a_\forall'$ and $a_\exists'$ have been maximized, as each of them receives the only item with a positive value.
    Item $g_\exists^-$ must be given to agent $a_\exists$, as only $a_\exists$ values it positively.
    Next, subject to those items $g_\exists$ and $g_\forall$ are taken by $a_\exists'$ and $a_\forall'$ respectively, consider the sum of the values of the remaining four agents over the allocation of the remaining $2n+3$ items.
    This sum is bounded by $W+(2Y_\exists+2Y_\forall+\sum_{i=1}^nx_i)+(20(Y_\exists+2Y_\forall+W)+2(Y_\exists+Y_\forall+T)-\frac13)$, where $W$ is the value contributed by item $g_\exists^-$, $(2Y_\exists+2Y_\forall+\sum_{i=1}^nx_i)$ is the value contributed by $c_1,\neg c_1,\ldots,c_n,\neg c_n$ and $(20(Y_\exists+2Y_\forall+W)+2(Y_\exists+Y_\forall+T)-\frac13)$ is the maximum possible value contributed by $g_{po},g_T,g_{poT}$.
    This sum for the current allocation is almost maximized: it is exactly this upper bound minus $\frac13$ due to that $g_T$ is allocated to $a_{po}$ instead of $a_\forall$.
    Therefore, if there were a Pareto-improvement, an agent's improvement in valuation is less than $1$.
    This is impossible for agent $a_\exists$ and $a_\forall$ since their valuations are integral.
    In addition, the valuations for these two agents must remain unchanged during the Pareto-improvement.
    For this to be possible for agent $a_{po}$ or $a_{po}'$, there must be a bundle with value exactly $10(Y_\exists+2Y_\forall+W)+(Y_\exists+Y_\forall+T)$ (notice that, in the current allocation, each of them receives value $10(Y_\exists+2Y_\forall+W)+(Y_\exists+Y_\forall+T-\frac13)$).
    The first term $10(Y_\exists+2Y_\forall+W)$ must come from the item $g_{po}$, and the second term $(Y_\exists+Y_\forall+T)$ must come from a subset of $\bigcup_{i=1}^nS_i$.
    Therefore, the only possible Pareto-improvement is to let agent $a_{po}$ replace item $g_T$ in her bundle by a subset of items from $\bigcup_{i=1}^nS_i$ with value exactly $(Y_\exists+Y_\forall+T)$.
    We will conclude the completeness of the reduction by showing that this is impossible.

    Let $A_\exists$ be the bundle received by $a_\exists$ with item $g_\exists^-$ excluded (in the current allocation before the Pareto-improvement), and we have $A_\exists\subseteq \bigcup_{i=m+1}^nS_i$.
    Consider the set of items $(\bigcup_{i=m+1}^nS_i)\setminus A_\exists$.
    By our construction, each of the four agents values this bundle $Y_\exists+\sum_{x_i\in U_\exists}x_i$.
    Let $A_\exists'$ be the bundle received by $a_\exists$ after the Pareto-improvement, with item $g_\exists^-$ excluded.
    Since we have seen that the value of $a_\exists$ cannot change at all, the value for $A_\exists$ and $A_\exists'$ must be the same for $a_\exists$, and this should also be true for the other three agents (as all the four agents have identical valuations for items in $\bigcup_{i=m+1}^nS_i$).
    Therefore, the value of $(\bigcup_{i=m+1}^nS_i)\setminus A_\exists'$ must still be $Y_\exists+\sum_{x_i\in U_\exists}x_i$ for all the four agents.
    In addition, for agent $a_{po}$, the subset of items used to replace $g_T$ must contain the whole set $(\bigcup_{i=m+1}^nS_i)\setminus A_\exists'$: if one item is missing from $(\bigcup_{i=m+1}^nS_i)\setminus A_\exists'$, the value of this subset will be less than $Y_\exists$ as $y_i$ for each $i=m+1,\ldots,n$ is larger than the value of the set of all items in $\bigcup_{i=1}^mS_i$.
    After taking over $(\bigcup_{i=m+1}^nS_i)\setminus A_\exists'$, we need to find a subset of items from $\bigcup_{i=1}^mS_i$ that is worth exactly $(Y_\exists+Y_\forall+T)-(Y_\exists+\sum_{x_i\in U_\exists}x_i)=Y_\forall+T-\sum_{x_i\in U_\exists}x_i$.
    To match the term $Y_\forall$, this subset must contain exactly one item from $S_i$ for each $i=1,\ldots,m$.
    However, this would imply there exists a subset $U_\forall\subseteq V_\forall$ such that $\sum_{x_j\in U_\forall}x_j=T-\sum_{x_i\in U_\exists}x_i$, contradicting to that the instance is a \yes instance.

    [Soundness] Suppose the \easubsetsum is a \no instance and suppose there exists an allocation $A$ that is both envy-free and Pareto-optimal.
    We will show that there must be a contradiction.

    Firstly, in $A$, the bundle allocated to one of $a_{po}$ and $a_{po}'$ must contain $\{g_{poT}\}$, and the bundle for the other one must contain $\{g_{po},g_T\}$.
    To see this, items $g_{po}$ and $g_{poT}$ are only valued by these two agents, and they cannot be put in the same bundle as $g_{po}$ is already more valuable than all items in $M\setminus\{g_{po},g_{poT}\}$.
    In addition, both agents' valuations must be exactly the same to keep envy-freeness, since they have identical valuation functions.
    The value for $g_{poT}$ is fractional, so the agent taking $g_{po}$ must further take $g_T$ (which is the only item with a fractional value, other than $g_{poT}$).
    We assume agent $a_{po}$'s bundle contains $\{g_{po},g_T\}$ and agent $a_{po}'$'s bundle contains $\{g_{poT}\}$ without loss of generality.
    
    Secondly, in $A$, agent $a_\exists'$ and $a_\forall'$ must take precisely the bundles $\{g_\exists\}$ and $\{g_\forall\}$ respectively, as $g_\exists$ and $g_\forall$ are only valued items for agents $a_\exists'$ and $a_\forall'$ respectively.
    In addition, $g_\exists^-$ must be in agent $a_\exists$'s bundle for Pareto-optimality.
    Next, to ensure that agent $a_\exists$ does not envy $a_\exists'$ and that agent $a_\forall$ does not envy $a_\forall'$, given that the item $g_T$ has been allocated to agent $a_{po}$, it is easy to see that 1) $a_\exists$ must take exactly one item from $S_i$ for each $i=m+1,\ldots,n$,
    and 2) $a_\forall$ must take exactly one item from $S_i$ for each $i=m+1,\ldots,n$ and all the two items from $S_i$ for each $i=1,\ldots,m$.

    However, this allocation is not Pareto-optimal.
    Given that the \easubsetsum instance is a \no instance, agent $a_\forall$ can choose a subset of items with value exactly $Y_\exists+Y_\forall+T$ and give this subset to agent $a_{po}$ in exchange for the item $g_T$.
    Agent $a_\forall$'s value remains unchanged, and agent $a_{po}$'s value is increased by $\frac13$.
\end{proof}

\begin{lemma}\label{lem:hard:21inf-po}
    \efpo-$[2,\mono,\infty]$ is $\sigma2$-hard.
\end{lemma}
\begin{proof}
    We reduce the problem from the complement of \aeCNF in Problem~\ref{problem:aecnf}, which is $\sigma2$-complete.
    Let $n$ be the number of variables with $n=n_\forall+n_\exists$ where $|V_\forall|=n_\forall$ and $|V_\exists|=n_\exists$.
    Let $V_\forall=\{x_1,\ldots,x_{n_\forall}\}$ and $V_\exists=\{y_1,\ldots,y_{n_\exists}\}$.
    Assume $n\geq 3$ without loss of generality.

    We construct a fair division instance with two agents and $2n+2$ items.
    For each variable $x_i\in V_{\forall}$ (resp., $y_i\in V_{\exists}$), we construct two items $X_i=\{x_i,\neg x_i\}$ (resp., $Y_i=\{y_i,\neg y_i\}$).
    There are two additional items $h$ and $g$.
    Let $X=\{x_1,\neg x_1,\ldots,x_{n_{\forall}},\neg x_{n_{\forall}}\}$ and $Y=\{y_1,\neg y_1,\ldots,y_{n_{\exists}},\neg y_{n_{\exists}}\}$.
    For each clause $C$ in $\phi$, we view it as a set of three items corresponding to the three literals in $(\bigcup_{i=1}^{n_\forall}X_i)\cup(\bigcup_{i=1}^{n_\exists}Y_i)$.
    
    The valuation of agent $1$ is defined as follows, where $\mathbb{I}(\cdot)$ is the indicator function.
    $$v_1(S)=\left\{\begin{array}{ll}
        0, & \mbox{if }S\cap X_i=\emptyset\mbox{ for some }X_i\mbox{ or } S\cap Y_i=\emptyset\mbox{ for some }Y_i\\
        \sum\limits_{i=1}^{n_{\forall}}2^i\cdot \mathbb{I}(x_i\in S) + 2^{2n}, & \mbox{else if }h\notin S, g\in S,\mbox{ and }S\cap C\neq\emptyset\mbox{ for each clause }C \\
        \sum\limits_{i=1}^{n_{\forall}}2^i\cdot \mathbb{I}(x_i\in S), & \mbox{else if }h\notin S\\
        v_1(S\setminus\{h\})+2^{n+1}, & \mbox{else if }h\in S\\
    \end{array}\right..$$

    Then we define the valuation of agent $2$.
    Note that to agent $2$, whether item $h$ belongs to the bundle does not make a difference to its value.
    $$v_2(S)=\left\{\begin{array}{ll}
        0, & \mbox{if }g\notin S, S\cap X_i=\emptyset\mbox{ for some }X_i,\mbox{ or } S\cap Y_i=\emptyset\mbox{ for some }Y_i\\
        \sum\limits_{i=1}^{n_{\forall}}2^{i+1}\cdot \mathbb{I}(x_i\in S) + 2, & \mbox{else if }S\cap C\neq\emptyset\mbox{ for each clause }C \\
        \sum\limits_{i=1}^{n_{\forall}}2^{i+1}\cdot \mathbb{I}(x_i\in S)+1, & \mbox{otherwise}\\
    \end{array}\right..$$

    We first list some necessary conditions for an allocation to be both envy-free and Pareto-optimal.
    \begin{enumerate}[leftmargin=*]
        \item For each variable $w\in V_{\forall}\cup V_{\exists}$, each bundle must contain exactly one of $\{w,\neg w\}$. Otherwise, there exists an agent $i\in\{1,2\}$ whose value is $0$.
        In this case, to guarantee Pareto-optimality, agent $3-i$ must receive a bundle that attains her maximum possible value.
        It is easy to verify by our construction that the allocation cannot be envy-free.

        \item Item $g$ must be allocated to agent $2$, for otherwise, the value of agent $2$ is $0$, and agent $1$ must receive her maximum possible value.
        Envy-freeness is also violated.

        \item Based on Pareto-optimality, item $h$ must be allocated to agent $1$.
    \end{enumerate}

    To show the completeness, assume that $(\phi,V_{\forall},V_{\exists})$ is a \yes instance.
    We will construct an allocation $(A_1,A_2)$ that is both envy-free and Pareto-optimal.
    Consider the Boolean assignment to $V_{\forall}$ such that $\phi$ cannot be satisfied under any assignment to $V_{\exists}$.
    Denote by $X^2$ the set of items corresponding to this assignment to $V_{\forall}$ and $X^1=X\setminus X^2$.
    Specifically, if $x_i\in V_{\forall}$ is assigned ``true'', let $x_i\in X^2$ and $\neg x_i\in X^1$; otherwise, let $\neg x_i\in X^2$ and $x_i\in X^1$.
    Let $Y^1=\{y_i|y_i\in V_{\exists}\}$ and $Y^2=Y\setminus Y^1$.
    Let the allocation $(A_1,A_2)$ be $(\{h\}\cup X^1\cup Y^1, \{g\}\cup X^2\cup Y^2)$.
    For agent $1$, as $h\in A_1$, we have $v_1(A_1)\ge 2^{n+1}\ge 2^{n_{\forall}+1}$.
    Since $\phi$ cannot be satisfied under the assignment to $V_{\forall}$, there exists some clause $C$ such that $A_2\cap C=\emptyset$, thus $v_1(A_2)\le 2^{n_{\forall}+1}$.
    For agent $2$, we have $v_2(A_1)=0$ and $v_2(A_2)>0$.
    Therefore, the allocation satisfies envy-freeness.
    We now show that the allocation is also Pareto-optimal.
    For the sake of contradiction, if it is Pareto-dominated by another allocation $(B_1,B_2)$, we need to guarantee that $v_1(B_1)>0$ and $v_2(B_2)>0$, thus we may assume that $g\in B_2$, $|B_j\cap X_i|=1$, and $|B_j\cap Y_i|=1$ for each $j\in\{1,2\}$.
    We may also assume that $h\in B_1$ as adding $h$ to agent $1$'s bundle will lead to an increase in agent $1$'s value.
    If $B_1\cap X=A_1\cap X$ (equivalently, $B_2\cap X=A_2\cap X$), we have $v_1(A_1)=v_1(B_1)$ and $v_2(A_2)=v_2(B_2)$ regardless of the allocation of the items in $Y$, as $\phi$ cannot be satisfied when fixing the assignment to the variables in $V_{\forall}$.
    Now assume that $A_1\cap X\neq B_1\cap X$.
    It is easy to verify that if agent $1$ gains an increase of $k$ in the value by changing from $A_1$ to $B_1$, i.e., $v_1(B_1)-v_1(A_1)=k\in\mathbb{Z}^+$, then agent $2$'s value will decrease due to $v_2(B_2)-v_2(A_2)\le-2k+1<0$.
    On the other hand, if $v_2(B_2)-v_2(A_2)=2k$ where $k\in\mathbb{Z}^+$ (or $v_2(B_2)-v_2(A_2)=2k+1$ if $B_2\cap C\neq\emptyset$ for all $C$ and this is not true for $A_2$), then $v_1(B_1)-v_1(A_1)=-k<0$.
    Therefore, no allocation will Pareto-dominate $(A_1,A_2)$.
    
    To show the soundness, assume that $(\phi,V_{\forall},V_{\exists})$ is a \no instance.
    If there exists an allocation $(A_1,A_2)$ that is both envy-free and Pareto-optimal, due to the above necessary conditions, it must be $h\in A_1$, $g\in A_2$, $|A_j\cap X_i|=1$, and $|A_j\cap Y_i|=1$ for each $j\in\{1,2\}$.
    Let $X^2=A_2\cap X$ be an arbitrary set of items in $X$ allocated to agent $2$.
    As there exists an assignment to the variables in $V_{\exists}$ to make $\phi$ true given any assignment to $V_{\forall}$, we can always select a set of items $Y^2\subset Y$ such that $(X^2\cup Y^2)\cap C\neq\emptyset$ for any clause $C$ under any $X^2$.
    In this case, agent $2$ will achieve the maximum possible value $\sum_{i=1}^{n_{\forall}}2^{i+1}\cdot \mathbb{I}(x_i\in X^2)+2$, which will Pareto-dominate all allocations where agent $2$ receives strictly less values when fixing $X^2$.
    However, the allocation is not envy-free to agent $1$, as $v_1(A_1)\le 2^{n+2}<2^{2n}\leq v_1(A_2)$.
    Therefore, no allocation simultaneously achieves envy-freeness and Pareto-optimality.
\end{proof}

\begin{lemma}\label{lem:hard:21infty}
    Both \efpo-$[2,\mono,\infty]$ and \efmsw-$[2,\mono,\infty]$ are $\delta2$-hard.
\end{lemma}
\begin{proof}
    Note that it has been shown that \efmsw-$[2,\mono,\infty]$ is $\delta2$-complete in~\citep{bouveret2008efficiency} (which does not imply the hardness of \efpo-$[2,\mono,\infty]$).
    We put our proof because our proof tackles the two problems together.
    
    We present a reduction from the $\delta2$-complete language \maxoddSAT (Problem~\ref{problem:maxoddSAT}).
    Given a \maxoddSAT instance $\phi$ with variables $x_1,\ldots,x_n$, we construct the following fair division instance with two agents.
    The set of items is $\{g,x_1,\neg x_1,x_2,\neg x_2,\ldots,x_n,\neg x_n\}$.
    Let $M_i=\{x_i,\neg x_i\}$.
    Let $W=2^{n+10}$.
    Agent $1$'s valuation function $v_1$ is defined as follows.
    We define $v_1(S)=0$ for the following scenarios:
    \begin{itemize}
        \item if $g\notin S$ or  $S\cap M_i=\emptyset$ for some $M_i$, set $v_1(S)=0$;
        \item if $|S\cap M_i|=1$ for all $M_i$, then $S$ corresponds to a Boolean assignment  (where $x_i\in S$ means $x_i$ is assigned ``true'' and $\neg x_i\in S$ means $x_i$ is assigned ``false''); set $v_1(S)=0$ if this is not a satisfying assignment of $\phi$.
    \end{itemize}
    Otherwise, there are two scenarios: 1) $|S\cap M_i|\geq1$ for all $M_i$ and $|S\cap M_i|=2$ for some $M_i$, and 2) $|S\cap M_1|=1$ for all $M_i$ and $S$ gives a satisfying assignment to $\phi$.
    In both cases, define
    $$v_1(S)=10W+\sum_{i=1}^n2^{n+5-i}\cdot\mathbb{I}(x_i\in S),$$
    where $\mathbb{I}(\cdot)$ is the indicator function.
    Agent $2$'s valuation function $v_2$ is defined as follows.
    If $S\cap M_i=\emptyset$ for some $i=1,\ldots,n$, then $v_2(S)=0$.
    Otherwise,
    $$v_2(S)=\left\{\begin{array}{ll}
        W+2, & \mbox{if }S\mbox{ is the set of all }2n+1\mbox{ items} \\
        W+1, & \mbox{if }\{g,\neg x_n\}\subseteq S\mbox{ and }|S\cap M_i|\geq1\mbox{ for some }i=1,\ldots,n\\
        W, & \mbox{otherwise}
    \end{array}\right..$$
    It is easy to verify that both $v_1$ and $v_2$ are monotone and can be computed in polynomial time.
    Since an allocation maximizing the social welfare is always Pareto-optimal, it suffices to show the followings:
    \begin{itemize}[leftmargin=*]
        \item[] Completeness: if $\phi$ is a \yes instance, an envy-free and social welfare optimal allocation exists;
        \item[] Soundness: if $\phi$ is a \no instance, there is no Pareto-optimal envy-free allocation.
    \end{itemize}

    [Completeness] Suppose $\phi$ is a \yes instance. 
    Let $A\subseteq\{x_1,\neg x_1,\ldots,x_n,\neg x_n\}$ be a subset of items with $A\cap M_i=1$ for all $M_i$ that represents a lexicographically optimal satisfying assignment of $\phi$.
    Consider the allocation $(A_1,A_2)$ where agent $1$ receives $A_1=\{g\}\cup A$ and agent $2$ receives the remaining items.
    Since $\phi$ is a \yes instance, we have $x_n\in A_1$ and $\neg x_n\in A_2$.
    The allocation is envy-free: agent $1$'s value on $A_2$ is $0$ since $g\notin A_2$; we have $v_2(A_2)=v_2(A_1)=W$ since each of $A_1$ and $A_2$ contains exactly one item from each $M_i$ and $\{g,\neg x_n\}\not\subseteq A_1$.
    This allocation's social welfare $11W+\sum_{i=1}^n2^{n+5-i}\cdot\mathbb{I}(x_i\in A)$ is also maximized.
    To see this, first of all, notice that $W>\sum_{i=1}^n2^{n+5-i}$.
    If $g\notin A_1$, we have $v_1(A_1)=0$ and the social welfare is less than $2W$, which is suboptimal.
    In addition, each of $A_1$ and $A_2$ must contain exactly one item from each $M_i$.
    Otherwise, one of $v_1(A_1)$ and $v_2(A_2)$ is $0$, and the social welfare is suboptimal.
    Subject to these, the part $A$ in $A_1=\{g\}\cup A$ represents a Boolean assignment of the $n$ variables, and the valuation of agent $1$ has already been maximized by our construction.
    The valuation of agent $2$ has also been maximized given that we have shown $g\in A_1$.
    Thus, this is a social welfare optimal allocation.

    [Soundness] Suppose $\phi$ is a \no instance.
    We show that any allocation $(A_1,A_2)$, even allowed to be partial, fails either envy-freeness or Pareto-optimality.
    Let $(A_1^*,A_2^*)$ be the allocation where $A_1^*=\emptyset$ and $A_2^*$ is the set of all items.
    In this allocation, the valuations of the two agents are $0$ and $W+2$ respectively.
    
    Firstly, if $A_2\cap M_i=\emptyset$ for some $M_i$, then $v_2(A_2)=0$.
    If $A_1\cap M_i\neq\emptyset$ for all $M_i$, then agent $2$ envies agent $1$.
    Otherwise, $v_1(A_1)=0$, and this allocation is Pareto-dominated by $(A_1^\ast,A_2^\ast)$.
    We assume from now on $A_2\cap M_i\neq\emptyset$ for all $M_i$.
    Next, we must have $g\in A_1$ and $A_1\cap M_i\neq\emptyset$ for all $M_i$.
    Otherwise, $v_1(A_1)=0$, and we need to allocate all items to agent $2$ to ensure $(A_1^*,A_2^*)$ does not Pareto-dominate $(A_1,A_2)$.
    However, this allocation is not envy-free, as $v_1(A_2)=10W+\sum_{i=1}^n2^{n+5-i}>0=v_1(A_1)$ (as $A_2$ contains some $M_i$, actually all $M_i$'s).
    Therefore, we assume from now on that $g\in A_1$ and each of $A_1$ and $A_2$ contains exactly one item from each $M_i$.
    In this case, $A_1$ gives a Boolean assignment to $\phi$.

    There are two possibilities that $\phi$ is a \no instance: $\phi$ is not satisfiable or the lexicographically optimal satisfying assignment of $\phi$ assigns ``false'' to $x_n$.
    In the former case, we must have $v_1(A_1)=0$ and $v_2(A_2)\leq W+1$, and $(A_1,A_2)$ is Pareto-dominated by $(A_1^\ast,A_2^\ast)$.
    In the latter case, if $A_1$ does not correspond to a lexicographically optimal assignment, then the allocation is not Pareto-optimal: switching $A_1$ to the one that represents a lexicographically optimal assignment (while keeping $g$ in $A_1$) increases the valuation of agent $1$ while keeping agent $2$'s valuation unchanged (which is $W$ before or after the change).
    If $A_1$ corresponds to a lexicographically optimal assignment, then $\neg x_n\in A_1$.
    We have $v_2(A_1)=W+1>W=v_2(A_2)$, violating envy-freeness.
\end{proof}

\begin{lemma}\label{lem:hard:212-po}
    \efpo-$[2,\mono,2]$ is $\theta2$-hard.
\end{lemma}
\begin{proof}
    We reduce this problem from \vcmember, which is $\theta2$-complete (Problem~\ref{problem:vcmember}).
    Assume that the \vcmember instance $(G=(V,E),x)$ satisfies $E\neq\emptyset$ (otherwise, it is a trivial instance).
    In addition, assume the instance satisfies 1) $|V|=2n+3$ is odd and 2) the minimum vertex cover uses at most $n$ vertices.
    This can be assumed without loss of generality, as we can create $|V|+3$ extra dummy isolated vertices.
    Given such an instance, we construct the following fair division instance.
    There are $|V|=2n+3$ items corresponding to the vertices, and we still use $x$ to represent the item corresponding to the vertex $x$ in the instance.
    Agent $1$'s valuation function $v_1$ is defined as follows:
    $$v_1(S)=\left\{\begin{array}{ll}
        2, & \mbox{if }|S|\geq n+3\mbox{ and }x\in S \\
        1, & \mbox{if }|S|\geq n+3\mbox{ and }x\notin S \\
        1, & \mbox{if }|S|= n+2 \\
        1, & \mbox{if }|S|\leq n+1\mbox{ and }S\mbox{ is a vertex cover}\\
        0, & \mbox{otherwise}
    \end{array}\right..$$
    Agent $2$'s valuation function is defined by $v_2(S)=|S|$.
    It is straightforward to check that both $v_1$ and $v_2$ are binary and polynomial time computable.

    If the \vcmember instance is a \yes instance, the allocation $(A_1,A_2)$ where $A_1$ is the minimum size vertex cover with $x\in A_1$ and $A_2=V\setminus A_1$ is both envy-free and Pareto-optimal.
    Envy-freeness is straightforward: $v_1(A_1)=1=v_1(A_2)$ since $x\notin A_2$, and $v_2(A_2)=|A_2|\geq n+3 >|A_1|=v_2(A_1)$.
    To show Pareto-optimality, notice that agent $1$'s value cannot be increased without decreasing agent $2$'s value, as increasing agent $1$'s value from $1$ to $2$ would require $|A_1|\geq n+3$, which reduces agent $2$'s value to at most $n$.
    On the other hand, agent $2$'s increment in value would reduce agent $1$'s value to $0$, as $A_1$ is a minimum size vertex cover.

    If the \vcmember instance is a \no instance, we show that any allocation $(A_1,A_2)$ fails either envy-freeness or Pareto-optimality.
    This is actually true even if partial allocations are allowed.
    Firstly, we must have $|A_2|\geq n+2$ to ensure agent $2$ does not envy agent $1$ and Pareto-optimality.
    This makes $|A_1|\leq n+1$.
    By our definition of $v_1$, $v_1(A_1)\leq 1$, and $v_1(A_1)=1$ if and only if $A_1$ is a vertex cover.
    If $A_1$ is not a vertex cover, then $v_1(A_1)=0$, and agent $1$ will envy agent $2$ since $|A_2|\geq n+2$ which implies $v_1(A_2)\geq 1$.
    If $A_1$ is not a minimum size vertex cover, then the allocation is not Pareto-optimal, since we can find a smaller size vertex cover $A_1'$ and the allocation $(A_1',A_2'=V\setminus A_1')$ Pareto-dominates $(A_1,A_2)$.
    If $A_1$ is a minimum size vertex cover, then $x\notin A_1$ since the \vcmember instance is a \no instance.
    We must have $x\in A_2$ (instead of throwing $x$ away) to maintain Pareto-optimality.
    Since a minimum vertex cover has size at most $n$, we have $|A_2|\geq n+3$.
    This would imply $v_1(A_2)=2$, and agent $1$ envies agent $2$.
\end{proof}

\begin{lemma}\label{lem:hard:212-msw}
    \efmsw-$[2,\mono,2]$ is $\theta2$-hard.
\end{lemma}
\begin{proof}
    We reduce this problem from \vcmember (Problem~\ref{problem:vcmember}).
    Given a \vcmember instance $(G=(V,E),x)$ with $n=|V|$ and $m=|E|$, we assume without loss of generality that a minimum vertex cover uses strictly less than $m$ vertices (by adding a triangle when necessary).
    We first construct a gadget with $n+2m+2$ items, where the first $n$ items correspond to the $n$ vertices in $V$, and there are $2m+2$ additional items.
    We then copy the gadget $m$ times and obtain a fair division instance with $2$ agents and $m(n+2m+2)$ items.
    Let $x$ be the item corresponding to the vertex $x\in V$ in the first gadget.
    The valuation function of agent $1$ is defined by $v_1(S)=\max\{f(S),g(S)\}$, where $f(S)$ and $g(S)$ are given by
    $$f(S)=\left\{\begin{array}{ll}
        \min\{\left\lfloor\frac{|S|}{2m}\right\rfloor, m+1\}, & \mbox{if }x\in S \\[1ex]
        \min\{\left\lfloor\frac{|S|}{2m}\right\rfloor, m\}, & \mbox{if }x\notin S
    \end{array}\right.$$
    and
    $$g(S)=\mbox{number of gadgets where the items from the first }n\mbox{ items represent a vertex cover}.$$
    The valuation function of agent $2$ is defined by $v_2(S)=\left\lfloor\frac{|S|}{m}\right\rfloor$.
    We may easily verify that the marginal value when adding a new item to a bundle is at most $1$ for both agents, thus the valuation functions satisfy general binary.

    Assume the minimum vertex cover in $G$ has size $k$ where $k<n$.
    We first show that an allocation $(A_1,A_2)$ has the maximum social welfare $n-k+3m+2$ if and only if agent $1$ receives the $k$ items corresponding to the minimum vertex cover from each gadget such that $v_1(A_1)=m$, and $A_2$ contains the remaining items such that $v_2(A_2)=n-k+2m+2$.
    This is proved by contradiction.
    Assume that in an allocation $(B_1,B_2)$, agent $1$ has valuation $t$ to $B_1$.
    If $t=f(B_1)$, we have $t\le m+1$ and $|B_1|\ge 2mt$.
    In this case, there are at most $m(n+2m-2t+2)$ items left for agent $2$, thus $v_2(B_2)\le n+2m-2t+2$.
    The social welfare of $(B_1,B_2)$ is at most $n+2m-t+2\le n+2m+2<n+3m-k+2$ as $m>k$, which is suboptimal.
    Therefore, we may assume that $t=g(B_1)$: there are $t$ gadgets whose graphs are fully covered, in each of which agent $1$ receives at least $k$ items corresponding to the minimum vertex cover.
    In this case, there are at most $(m-t)(n+2m+2)+t(n+2m+2-k)$ items left for agent $2$, thus
    $$v_2(B_2)\le \left\lfloor\frac{(m-t)(n+2m+2)+t(n+2m+2-k)}{m}\right\rfloor\le n+2m+2-\frac{kt}{m}.$$
    Only when $t=m$, the social welfare of $(B_1,B_2)$ reaches the optimal $n-k+3m+2$, and we conclude that agent $1$ needs to receive the items corresponding to a vertex cover in each gadget to obtain a maximum social welfare.
    On the other hand, if the items that agent $1$ receives in a gadget correspond to a vertex cover, yet is not minimum, the value of agent $1$ will still remain $m$, yet the number of items left for agent $2$ will be at most $m(n+2m+2-k)-1$, leading to a suboptimal social welfare of at most $n-k+3m+1$.

    Given the above fact, it is straightforward to see that an allocation $(A_1,A_2)$ that both maximizes social welfare and satisfies envy-freeness exists if and only if there is a minimum vertex cover in $G$ that contains $x$.
    When \vcmember is a \yes instance, we let $A_1$ contain $mk$ items where each $k$ items correspond to a minimum vertex cover that contains $x$ for each gadget.
    Agent $2$ does not envy agent $1$ as $|A_2|>|A_1|$, and agent $1$ does not envy agent $2$ as $v_1(A_1)=v_1(A_2)=m$.
    When \vcmember is a \no instance, each minimum vertex cover in $G$ does not contain $x$, and $v_1(A_1)=m<m+1=v_1(A_2)$ as $x\in A_2$ and $|A_2|\ge m(2m+2)$.
    Envy-freeness from agent $1$ to agent $2$ will be violated.
\end{proof}

\subsection{Combining Together}
\label{sect:combine}
We index all our results by Table~\ref{tab:resultsIndices}, and relate them to the theorems below (which are straightforward combinations of results in Section~\ref{sec:infty0k-po} and the previous two sub-sections).

\begin{table}[h]
    \centering
    \setlength{\tabcolsep}{3pt}
    \renewcommand\arraystretch{1.1}
    \begin{tabular}{|l||ccc|ccc|}
    \hline
       \multirow{2}{*}{\textbf{EF+MSW}}  & \multicolumn{3}{c|}{Additive Valuations} & \multicolumn{3}{c|}{General Monotone Valuations} \\
         & binary & $k$-ary & general & binary & $k$-ary & general\\
    \hline
    $n=2$ & Result 1 & Result 2 & Result 3 & Result 4 & Result 5 & Result 6\\
    constant $n$ & Result 7 & Result 8 & Result 9 & Result 10 & Result 11 & Result 12\\
    general $n$ & Result 13 & Result 14 & Result 15 & Result 16 & Result 17 & Result 18\\
    \hhline{=||======}
    \multirow{2}{*}{\textbf{EF+PO}} & \multicolumn{3}{c|}{Additive Valuations} & \multicolumn{3}{c|}{General Monotone Valuations} \\
         & binary & $k$-ary & general & binary & $k$-ary & general\\
     \hline
    $n=2$ & Result 19 & Result 20 & Result 21 & Result 22 & Result 23 & Result 24\\
    constant $n$ & Result 25 & Result 26 & Result 27 & Result 28 & Result 29 & Result 30\\
    general $n$ & Result 31 & Result 32 & Result 33 & Result 34 & Result 35 & Result 36\\
    \hline
    \end{tabular}
    \caption{Results Indices}
    \label{tab:resultsIndices}
\end{table}

\begin{theorem}[Results 1, 2, 7, 8, 19, 20, 25, 26]
    For any constant $n\geq 2$ and $k\geq 2$, \efpo-$[n,\add,k]$ and \efmsw-$[n,\add,k]$ are in $\classP$.
\end{theorem}
\begin{proof}
    This is already proved in Lemma~\ref{lem:member:n0k-po}.
\end{proof}

\begin{theorem}[Results 13, 14, 15]\label{thm:infty0x-msw}
    \efmsw-$[\infty,\add,\cdott]$ is $\classNP$-complete.
\end{theorem}
\begin{proof}
    The problem is in $\classNP$ by Lemma~\ref{lem:member:infty0infty-msw}.
    To show its $\classNP$-hardness, the proof of Proposition~21 in~\citep{bouveret2008efficiency} applies here. 
    In particular, they show \efpo is $\classNP$-hard under additive binary valuations.
    For these valuations, Pareto-optimality and social welfare optimality are equivalent, as both require every item to be allocated to some agent that positively values it. 
    Hence, \efmsw under additive binary valuations is also $\classNP$-hard.
\end{proof}

\begin{theorem}[Results 3, 9, 15]\label{thm:x0infty-msw}
    \efmsw-$[\cdott,\add,\infty]$ is $\classNP$-complete.
\end{theorem}
\begin{proof}
    The problem is in $\classNP$ by Lemma~\ref{lem:member:infty0infty-msw}.
    A straightforward reduction to show its $\classNP$-hardness is through the $\classNP$-complete problem \partition, where each element in the \partition instance represents an item and the two agents have identical valuation functions. 
    Every allocation maximizes social welfare and there exists an envy-free allocation if and only if the \partition instance is a \yes instance.
\end{proof}

\begin{theorem}[Results 4, 5, 10, 11, 16, 17]\label{thm:x1k-msw}
    For any constant $k\ge2$, \efmsw-$[\cdott,\mono,k]$ is $\theta2$-complete.
\end{theorem}
\begin{proof}
    It directly follows from Lemma~\ref{lem:member:infty1k-msw} and Lemma~\ref{lem:hard:212-msw}.
\end{proof}

\begin{theorem}[Results 6, 12, 18]\label{thm:x1infty-msw}
    \efmsw-$[\cdott,\mono,\infty]$ is $\delta2$-complete.
\end{theorem}
\begin{proof}
    The result is given by~\citet{bouveret2008efficiency}.
    We provide an alternative proof in Lemma~\ref{lem:hard:21infty}.
\end{proof}

\begin{theorem}[Results 22, 23, 28, 29]\label{thm:n1k-po}
    For any constant $n\ge 2$ and $k\ge 2$, \efpo-$[n,\mono,k]$ is $\theta2$-complete.
\end{theorem}
\begin{proof}
    It directly follows from Lemma~\ref{lem:member:n1k-po} and Lemma~\ref{lem:hard:212-po}.
\end{proof}

\begin{theorem}[Results 34, 35, 36]\label{thm:infty1x-po}
    \efpo-$[\infty,\mono,\cdott]$ is $\sigma2$-complete.
\end{theorem}
\begin{proof}
    The result is given by~\citet{bouveret2008efficiency}.
\end{proof}

\begin{theorem}[Results 24, 30, 36]\label{thm:x1infty-po}
    \efpo-$[\cdott,\mono,\infty]$ is $\sigma2$-complete.
\end{theorem}
\begin{proof}
    Noting that the containment in $\sigma2$ is trivial, it directly follows from Lemma~\ref{lem:hard:21inf-po}.
    Notice that the $\sigma2$-completeness (Result 36) is also implied by~\citet{de2009complexity}.
\end{proof}

\begin{theorem}[Result 21]
    \efpo-$[2,\add,\infty]$ is $\classNP$-complete.
\end{theorem}
\begin{proof}
    The problem is in $\classNP$ by Lemma~\ref{lem:member:20infty-po}.
    It is $\classNP$-complete for the same reason as it is in the proof of Theorem~\ref{thm:x0infty-msw}.
\end{proof}

\begin{theorem}[Result 31]
    \efpo-$[\infty,\add,2]$ is $\classNP$-complete.
\end{theorem}
\begin{proof}
    The result is given by~\citet{bouveret2008efficiency}.
\end{proof}

\begin{theorem}[Results 32, 33]
    \efpo-$[\infty,\add,k]$ is $\sigma2$-complete for any $k\ge 3$.
\end{theorem}
\begin{proof}
    Noting that the containment in $\sigma2$ is trivial, it directly follows from Lemma~\ref{lem:hard:infty03-po}.
\end{proof}

\begin{theorem}[Results 27, 33]
    \efpo-$[n,\add,\infty]$ is $\sigma2$-complete for any $n\geq 6$.
\end{theorem}
\begin{proof}
    The containment in $\sigma2$ is trivial.
    The hardness is by Lemma~\ref{lem:hardness:60infty}.
    Notice that the $\sigma2$-completeness for \efpo-$[\infty,\add,\infty]$ is also known by~\citet{de2009complexity}.
\end{proof}

\begin{remark}\label{rmk:partial}
    All results in our paper continue to hold if we consider the space of allocations that allows partial allocations.
    All lemmas in Sect.~\ref{sect:member} hold with exactly the same proofs except for Lemma~\ref{lem:member:infty0infty-msw} and Lemma~\ref{lem:member:20infty-po}.
    These two lemmas are regarding additive valuations.
    For additive valuations, a Pareto-optimal allocation or a social welfare optimal allocation must allocate all items except for those that are valued $0$ for all agents.
    These items can be discarded without loss of generality.
    The same arguments in the proofs of the two lemmas continue to work if we remove these items from the item-set.

    For the hardness results, for additive valuations, we have never constructed an item that is valued $0$ for all agents. Therefore, every efficient allocation must be complete, and the proofs continue to work.
    For general monotone valuations, we have made explicit in the reduction that 1) for the \yes instance, a complete fair and efficient allocation exists, and 2) for the \no instance, no valid allocation exists even allowing partial allocations.

\section{Proofs of Results in Table~\ref{tab:results2}}\label{sec:efficiency-checking-proof}
Note that for efficiency checking, all the above-mentioned problems under different parameters belong to $\classcoNP$.
Thus, we only show the hardness in the following.

The problem of checking whether an allocation maximizes social welfare can be solved in polynomial time when the valuations are additive, simply by checking whether each item is allocated to the agent with a maximum value for it.
On the other hand, the problem becomes $\classcoNP$-hard when the valuations are general monotone, even for two agents with binary valuations.

\begin{lemma}\label{lem:decidingMSW-mono}
    Deciding whether an allocation maximizes social welfare is $\classcoNP$-hard under general monotone valuations, even for two agents with binary valuations.
\end{lemma}
\begin{proof}
    The problem can be reduced from the $3$SAT problem.
    Given a $3$SAT instance with $n$ variables, we construct two items $\{x,\neg x\}$ from each variable $x$.
    There are two agents $1$ and $2$.
    Agent $1$'s valuation is $1$ only when she receives at least one item constructed from each variable, and her bundle covers each clause.
    That is, for each clause with variables $x_1,x_2,$ and $x_3$, agent $1$ needs to receive at least one item among $x_1,x_2,$ and $x_3$.
    Otherwise, her value is $0$.
    For agent $2$, if she receives either one item constructed from one variable, her utility increases by $1$.
    The increase is still $1$ if she receives both items constructed from one variable.
    Suppose we are given an allocation where agent $1$ has utility $0$ and agent $2$ has utility $n$.
    Deciding whether this allocation with a social welfare of $n$ maximizes the social welfare is equivalent to deciding whether the $3$SAT instance is unsatisfiable.
\end{proof}

For Pareto-optimality, we first focus on additive valuations.
We know the problem belongs to $P$ when the valuations are $k$-ary and the number of agents is constant.
\citet{aziz2019efficient} shows that the problem is $\classcoNP$-complete for tri-valued valuations (each value belongs to $\{p,q,r\}$ for some positive numbers $p,q,r$).
We extend their result to $3$-ary valuations where each value is at most $2$.

We first present a proof for $4$-ary valuations, which conveys the key insight of our construction of the Pareto-improvement cycle.
\begin{lemma}\label{lem:decidingPO4}
    Deciding whether an allocation is Pareto-optimal is $\classcoNP$-hard for $k=4$-ary additive valuations.
\end{lemma}
\begin{proof}
    We reduce the problem from the $\classNP$-complete problem \krs, which decides the existence of a $\kappa$-regular subgraph in a given bipartite graph for $\kappa\ge 3$.
    In particular, we consider the case with $\kappa=3$.
    We will show that an allocation is Pareto-optimal if and only if there is no $3$-regular subgraph in $G$.

    Given an \krs instance $G=(U\cup V, E)$, assume that $|U|=|V|=t$, for otherwise we may add dummy vertices with no incident edges.
    Let $U=\{u_1,\ldots,u_t\}$ and $V=\{v_1,\ldots,v_t\}$.
    We construct an allocation ${A}=(A_1,\ldots,A_n)$ with $n=2t+2$ agents and $m=3t+|E|$ items.
    Specifically, for each vertex $u_i\in U$, we construct an agent $a_i$ with $A_{a_i}=\{g_1^{a_i},\ldots,g_d^{a_i}\}$, where $d$ is the degree of $u_i$ in $G$.
    For each vertex $v_i\in V$, we construct an agent $b_i$ with $A_{b_i}=\{g^{b_i}\}$.
    There are two additional agents $c$ and $d$, with $A_c=\{g_1^c,\ldots,g_{t-1}^c,g_{po}\}$ and $A_d=\{g_1^d,\ldots,g_t^d\}$.
    
    The agents' valuations are defined as follows.
    Note that by our construction, each item is positively valued by exactly two agents, and each agent's value to each item belongs to $\{0,1,2,3\}$.
    \begin{itemize}
        \item Each agent $a_i$ where $u_i\in U$ has value $1$ to each item in $A_{a_i}$, and value $3$ to item $g^d_i$.
        
        \item Each agent $b_i$ where $v_i\in V$ has value $3$ to item $g^{b_i}$.
        For each edge $e=(u_j, v_i)\in E$ for some $u_j$, assume $e$ is the $\ell$-th edge incident to $u_j$ in $G$, then $b_i$ has value $1$ to item $g_{\ell}^{a_j}\in A_{a_j}$.
        
        \item Agent $c$ has value $2$ to item $g_{po}$, and $3$ to each item in $A_c\setminus\{g_{po}\}$ as well as each item $g^{b_i}$ for $i\in\{1,\ldots,t\}$.
        
        \item Agent $d$ has value $3$ to each item in $A_d\cup A_c$.
    \end{itemize}

    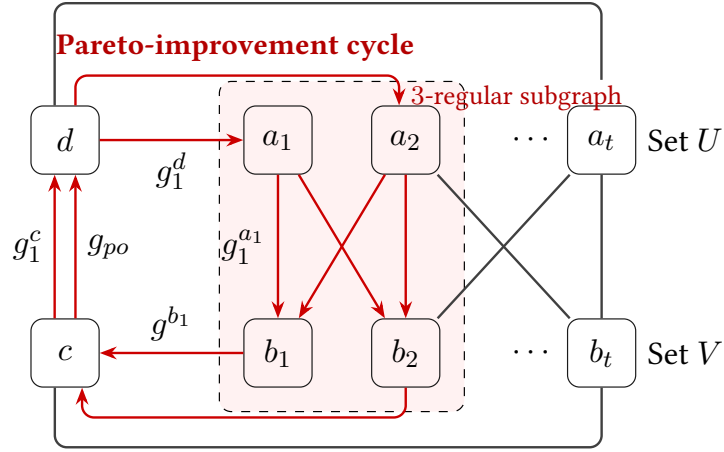
\begin{figure}[h]
    \centering
    \resizebox{0.6\textwidth}{!}
    {
    \begin{tikzpicture}[
        node distance=1.5cm and 3cm,
        agent/.style={rectangle, draw=black!90, rounded corners, minimum size=0.8cm, font=\small, align=center, font=\bfseries},
        item_transfer/.style={->, >={Stealth[length=2mm]},  thick, color=red!80!black,rounded corners=4pt},
        edge/.style={-, thick, color=gray!50!black,rounded corners=4pt},
        label_text/.style={font=\itshape, align=center, color=black},
        label_text2/.style={font=\itshape, align=center, color=black, font=\small},
        section_box/.style={draw, dashed, rounded corners, inner sep=8pt},
        section_box2/.style={draw, dashed, rounded corners, inner sep=8pt}
    ]

    \node[agent] (d) at (-2.5, 2.5) {$d$};
    \node[agent] (c) at (-2.5, 0) {$c$};

    \node[agent] (a1) at (0,2.5) {$a_1$};
    \node[agent] (a2) at (1.5,2.5) {$a_2$};
    \node[agent,label=left:{$\cdots$},label=right:{Set $U$}] (at) at (3.8,2.5) {$a_{t}$};

    \node[agent] (b1) at (0,0) {$b_1$};
    \node[agent] (b2) at (1.5,0) {$b_2$};
    \node[agent,label=left:{$\cdots$},,label=right:{Set $V$}] (bt) at (3.8,0) {$b_{t}$};

    \draw[item_transfer] (a1) -- (b1) node[midway, left, label_text] {$g^{a_1}_1$};
    \draw[item_transfer] (a1) -- (b2) node[midway, left, label_text] {};
    \draw[item_transfer] (a2) -- (b2) node[midway, left, label_text] {};
    \draw[item_transfer] (a2) -- (b1) node[midway, left, label_text] {};
    \draw[edge] (a2) -- (bt);
    \draw[edge] (at) -- (b2);
    \draw[edge] (at) -- (bt);

    \draw[item_transfer] (b1) -- node[midway, above, label_text] {$g^{b_1}$} (c);
    \draw[item_transfer] (b2.south) -- ++(0,-0.35) -- ++(-3.8,0) -- ++(0,0.35);
    \draw[edge] (bt.south) -- ++(0,-0.7) -- ++(-6.42,0) -- ++(0,0.7);
    \draw[item_transfer] ($(c.north west)!0.35!(c.north east)$) -- node[midway, left, label_text] {$g_{1}^c$} ($(d.south west)!0.35!(d.south east)$);
    \draw[item_transfer] ($(c.north west)!0.65!(c.north east)$) -- node[midway, right, label_text] {$g_{po}$} ($(d.south west)!0.65!(d.south east)$);
    \draw[item_transfer] (d) -- node[midway, below, label_text] {$g^{d}_1$} (a1);
    \draw[item_transfer] ($(d.north west)!0.65!(d.north east)$) -- ++(0,0.35) -- ++(3.8,0) -- ++(0,-0.35);
    \draw[edge] ($(d.north west)!0.35!(d.north east)$) -- ++(0,1.2) -- ++(6.42,0) -- ++(0,-0.9);

    \node[color=red!70!black] at (2.8, 3) {\footnotesize $3$-regular subgraph};
    \node[color=red!70!black] at (-0.5, 3.6) {\small\textbf{Pareto-improvement cycle}};

    \begin{scope}[on background layer]
        \node[section_box, fill=red!5, fit=(a1)(b2)] {};
    \end{scope}

    \end{tikzpicture}
    }
    \caption{Visualization of construction in Lemma~\ref{lem:decidingPO4}. Vertices represent agents and edges represent items. Items are only positively valued by endpoint agents. The directed red edges form a Pareto-improvement cycle, on which each $a_i$ gives away $3$ edge items and each $b_i$ receives $3$ edge items. This enforces the $3$-regular subgraph constraint. Agent $c$'s utility strictly increases by $1$ by receiving an item $g^{b_i}$ with value $3$ and giving $g_{po}$ with value $2$ to $d$.}
    \end{figure}

    Given the above construction, if the graph $G$ contains a $3$-regular subgraph $G'=(U'\cup V',E')$ with $|U'|=|V'|=t'$, we construct the following Pareto-improvement based on ${A}$.
    For each $u_i\in U'$, let $a_i$ receive $g_i^d$ from $A_d$.
    For each $v_i\in V'$ with edge $e=(u_j,v_i)\in E'$, let $b_i$ receive the corresponding item with value $1$ from $A_{a_j}$ following the edge $e$.
    Agent $c$ receives all items $g^{b_i}$ for $v_i\in V'$, and agent $d$ receives $\{g_1^c,\ldots,g_{t'-1}^c,g_{po}\}$ from $A_c$.
    It is easy to verify that agent $c$'s utility strictly increases by $1$, and all other agents' utilities remain the same.

    On the other hand, assume that $G$ does not contain a $3$-regular subgraph, we will show that ${A}$ is Pareto-optimal by contradiction.
    Suppose there exists an allocation ${B}=(B_1,\ldots,B_n)$ that Pareto-dominates ${A}$.
    We assume that ${B}$ is non-wasteful, that is, each agent only receives items that she values positively, for otherwise we may reallocate the ``wasted'' items without decreasing any agent's utility.
    The only way to increase social welfare is to let agent $d$ receive $g_{po}$ from $A_c$.
    Let $x=|B_c\setminus A_c|$ and $y=|A_d\setminus B_d|$, where $x$ and $y$ respectively represent the number of items transferred from agents $\{b_1,\ldots,b_t\}$ to agent $c$, and from agent $d$ to agents $\{a_1,\ldots,a_t\}$.
    It is guaranteed that $x\ge y$, for otherwise the social welfare of agents $c$ and $d$ will decrease under any allocation, therefore ${B}$ cannot be a Pareto-improvement.
    Similarly, we have $x\le y$ by considering the social welfare of the remaining agents.
    Combining the two inequalities, we have $x=y$.
    Therefore, there are $x$ agents in $\{b_1,\ldots,b_t\}$ that give away their items to agent $c$, among which each of them in turn needs to receive $3$ items with value $1$ from $\{a_1,\ldots,a_t\}$ to guarantee her utility does not decrease.
    There are also $y$ agents in $\{a_1,\ldots,a_t\}$ that receive items from agent $d$, and each may give away at most three items to $\{b_1,\ldots,b_t\}$.
    This corresponds to a $3$-regular subgraph in $G$, which contradicts the assumption and implies no Pareto-improvement is possible.
\end{proof}

We now turn to the case of $k=3$.
\begin{lemma}\label{lem:decidingPO3}
    Deciding whether an allocation is Pareto-optimal is $\classcoNP$-hard for $k=3$-ary additive valuations.
\end{lemma}
\begin{proof}
    We reduce the problem from the $\classNP$-complete problem \krs.
    In particular, we consider the case with $\kappa=4$.
    Given an \krs instance $G=(U\cup V, E)$, assume that $|U|=|V|=t$.
    We construct an allocation ${A}=(A_1,\ldots,A_n)$ with $n=4t+2$ agents and $m=7t+|E|$ items, and we will show that the allocation is Pareto-optimal if and only if $G$ does not contain a $4$-regular subgraph.
    
    For each vertex $u_i\in U$, we create two agents $a_i$ and $a'_i$.
    Let $A_{a_i}=\{g_1^{a_i},\ldots,g_d^{a_i}\}$ where $d$ is the degree of $u_i$ in $G$.
    Let $A_{a'_{i}}=\{g^{a'_i}_1,g^{a'_i}_2\}$.
    For each vertex $v_i\in V$, we create two agents $b_i$ and $b'_i$.
    Let $A_{b_i}=\{g_1^{b_i},g_2^{b_i}\}$ and $A_{b'_i}=\{g^{b'_i}\}$.
    There are two additional agents $c$ and $d$, with $A_{c}=\{g_{1}^c,\ldots,g_{t-1}^c,g_{po}\}$ and $A_{d}=\{g_{1}^d,\ldots,g_{t}^d\}$.
    The agents' valuations are defined as follows.
    Note that by our construction,
    each agent's value to each item is either $0$, $1$, or $2$.
    \begin{itemize}
        \item Each agent $a_i$ where $u_i\in U$ has value $1$ to each item in $A_{a_i}$ and $2$ to each item in $A_{a'_{i}}$.
        \item Each agent $a'_{i}$ where $u_i\in U$ has value $1$ to each item in $A_{a'_{i}}$ and $2$ to each item in $A_{d}$.
        \item Each agent $b_i$ where $v_i\in V$ has value $2$ to each item in $A_{b_i}$.
        For each edge $e=(u_j,v_i)\in E$, assume $e$ is the $\ell$-th edge incident to $u_j$ in $G$, then $b_i$ has value $1$ to item $g^{a_j}_\ell\in A_{a_i}$.
        \item Each agent $b'_i$ where $v_i\in V$ has value $2$ to each item in $A_{b'_i}$ and $1$ to each item in $A_{b_i}$.
        \item Agent $c$ has value $1$ to item $g_{po}$, and $2$ to each item in $A_c\setminus\{g_{po}\}$ as well as each item $g^{b'_i}$ for $i\in\{1,\ldots,t\}$.
        \item Agent $d$ has value $2$ to each item in $A_{d}\cup A_c$.
    \end{itemize}

    Given the above construction, if the graph $G$ contains a $4$-regular subgraph $G'=(U',V',E')$ with $|U'|=|V'|=t'$, we construct the following Pareto-improvement based on ${A}$.
    For each $u_i\in U'$, let $a'_{i}$ receive one distinct item from $A_{d}$, and let $a_i$ receive the two items in $A_{a'_{i}}$.
    For each $v_i\in V'$ with edge $e=(u_j,v_i)\in E'$, let $b_i$ receive the corresponding item with value $1$ from $A_{a_j}$ following the edge $e$, and let $b'_{i}$ receive the two items in $A_{b_i}$.
    Finally, agent $c$ further receives all items in $\bigcup_{v_i\in V'}A_{b'_i}$, and gives away $\{g^c_1,\ldots,g^c_{t'-1},g_{po}\}$ to agent $d$.
    It is easy to verify that agent $c$'s utility strictly increases by $1$, and all other agents' utility remains the same.

    On the other hand, if $G$ does not contain a $4$-regular subgraph, we show that ${A}$ is Pareto-optimal by contradiction.
    Suppose there exists an allocation ${B}=(B_1,\ldots,B_n)$ that Pareto-dominates ${A}$.
    Without loss of generality, we may assume that each agent only receives items that she values positively.
    Let $S_1=B_{c}\setminus A_{c}, S_2=A_d\setminus B_d$, and let $x=|S_1|, y=|S_2|$.
    It is guaranteed that $x\ge y$, for otherwise the social welfare of agents $c$ and $d$ will decrease under any allocation, thus ${B}$ cannot be a Pareto-improvement.
    We now consider the social welfare of the remaining agents.
    Let $p$ be the total number of items transferred from agent $a_i'$ to the corresponding agent $a_i$ for each $u_i\in U$, and $q$ be that transferred from agent $b_i$ to the corresponding agent $b_i'$ for $v_i\in V$.
    Only an agent $a'_i$ with $B_{a'_i}\cap S_2\neq\emptyset$ may gives the two items $g^{a_i'}_{1}$ and $g_{2}^{a_i'}$ to $a_i$, therefore, $p\le 2y$.
    Each agent $b'_i$ with $g^{b_i'}\in S_1$ must take both items in $A_{b_i}$, therefore, $q\ge 2x$.
    The social welfare change for agents $\{a_i:u_i\in U\}$ and $\{b_i:v_i\in V\}$ is $2p-2q$, since the transfer of edge items preserves the total utility.
    This gives us $2x\le q\le p\le 2y$.
    Putting together, we have $x=y$ and $p=q=2x$.
    It is then implied that, each agent $a_i'$ receiving items from $S_2$ must give the two items to $a_i$, each such agent $a_i$ must give four edge items to agents in $\{b_i:v_i\in V\}$, and each such agent $b_i$ must receive four edge items to give the two items to $b_i'$.
    This corresponds to a $4$-regular subgraph in $G$, which contradicts the assumption and implies no Pareto-improvement is possible.
\end{proof}

When considering the case with two agents and additive valuations, although an envy-free allocation can always be improved to a Pareto-optimal one (if it was not) with envy-freeness maintained, finding such a Pareto-improvement is hard.
In fact, it is $\classNP$-complete to decide if such a Pareto-improvement is possible.
\begin{lemma}\label{lem:decidingPO2agents}
    Deciding whether an allocation is Pareto-optimal is $\classcoNP$-hard for two agents with additive valuations.
\end{lemma}
\begin{proof}
Given a subset-sum instance $(\{a_1,\ldots,a_n\},k)$, construct a fair division instances with $n+2$ items where the first $n$ items $g_1,\ldots,g_n$ corresponds to $a_1,\ldots,a_n$, and there are two more items $h_k$ and $h_{-k}$.
Both agents have values $a_1,\ldots,a_n$ for $g_1,\ldots,g_n$ respectively.
Agent $1$ has value $k+\varepsilon$ and $-k-\varepsilon+\sum_{i=1}^na_i$ for $h_k$ and $h_{-k}$ respectively, for some sufficiently small $\varepsilon>0$ (rescale the valuations to make valuations integral).
Agent $2$ has value $k$ and $-k+\sum_{i=1}^na_i$ for $h_k$ and $h_{-k}$ respectively.
Consider the envy-free allocation $(A_1=\{g_1,\ldots,g_n\},A_2=\{h_k,h_{-k}\})$.
It is easy to see that this allocation is Pareto-optimal if and only if the subset-sum instance is a \no instance: for a \yes subset-sum instance, agent $1$ can give a subset of items with total value exactly $k$ to agent $2$, in exchange of $h_k$ in agent $2$'s bundle.
\end{proof}

Finally, when considering monotone valuations, we can use the same proof as in Lemma~\ref{lem:decidingMSW-mono} to show the $\classcoNP$-hardness of Pareto-optimality checking.
\begin{lemma}
    Deciding whether an allocation is Pareto-optimal is $\classcoNP$-hard under general monotone valuations, even for two agents with binary valuations.
\end{lemma}
\end{remark}

\newpage
\bibliographystyle{ACM-Reference-Format}
\bibliography{ref}

\appendix

\end{document}